\documentclass{article}

\usepackage{amsthm,amsmath,amssymb,amsfonts}

\usepackage[square,numbers]{natbib}
\usepackage{rotating}
\usepackage{fancyvrb}
\usepackage{braket}
\usepackage{algorithm}
\usepackage{algpseudocode}
\usepackage{qcircuit}
\usepackage{authblk}
\usepackage{color}
\usepackage{hyperref, cleveref}

\providecommand{\keywords}[1]
{
  \small	
  \textbf{\textit{Keywords---}} #1
}

\newtheorem{theorem}{Theorem}[section]
\newtheorem{lemma}{Lemma}[section]

\newtheorem{corollary}{Corollary}[section]

\newtheorem{assumption}{Assumption}[section]
\newtheorem{remark}{Remark}[section]
\newtheorem{definition}{Definition}[section]

\title{An Efficient Quantum Algorithm for Linear System Problem in Tensor Format}

\author[1,2]{Zeguan Wu}
\author[2]{Sidhant Misra}
\author[1]{Tam\'as Terlaky}
\author[1]{Xiu Yang}
\author[2]{Marc Vuffray}

\affil[1]{Lehigh University}
\affil[2]{Los Alamos National Laboratory}

\date{}

\begin{document}

\maketitle

\begin{abstract}
    Solving linear systems is at the foundation of many algorithms. Recently, quantum linear system algorithms (QLSAs) have attracted great attention since they converge to a solution exponentially faster than classical algorithms in terms of the problem dimension. However, low-complexity circuit implementations of the oracles assumed in these QLSAs constitute the major bottleneck for practical quantum speed-up in solving linear systems. In this work, we focus on the application of QLSAs for linear systems that are expressed as a low rank tensor sums, which arise in solving discretized PDEs. Previous works uses modified Krylov subspace methods to solve such linear systems with a per-iteration complexity being polylogarithmic of the dimension but with no guarantees on the total convergence cost. We propose a quantum algorithm based on the recent advances on adiabatic-inspired QLSA and perform a detailed analysis of the circuit depth of its implementation. We rigorously show that the \emph{total complexity} of our implementation is polylogarithmic in the dimension, which is comparable to the \emph{per-iteration complexity} of the classical heuristic methods.
\end{abstract}

\keywords{quantum computing; linear system; tensor format}

\section{Introduction}\label{sec: intro}
Quantum computing has obtained great attention for its potential in solving certain problems faster than classical algorithms. Starting from Deutsch's theory, see \cite{deutsch1992rapid}, a large number of quantum algorithms have been proposed, including Shor's algorithm for integer factorization, \cite{shor1994algorithms}, Grover's algorithm for Database Search, \cite{grover1996fast}, Quantum Approximate Optimization Algorithm (QAOA) for combinatorial optimization problems, \cite{farhi2014quantum}, QLSAs for solving quantum linear system problem (QLSP), \cite{harrow2009quantum,childs2017quantum,chakraborty2018power,subacsi2019quantum,costa2022optimal,an2022quantum,jennings2023efficient}, and many other algorithms. We refer readers to the survey paper \cite{dalzell2310quantum} for more details.
Among these quantum algorithms, some provide polynomial speed-up while some provide exponential speed-up. QLSAs are a family of algorithms that show exponential convergence speed-up in terms of problem dimension when compared with classical methods including Cholesky factorization, \cite{harrow2009quantum,childs2017quantum,chakraborty2018power,subacsi2019quantum,costa2022optimal,an2022quantum,jennings2023efficient}.
Due to the significance of solving linear systems in numerical computation, a significant amount of research have been invested in using QLSAs to speed up classical algorithms. These efforts include using QLSAs to speed up classical PDE algorithms \cite{liu2021efficient,krovi2023improved,jennings2023cost} and classical optimization algorithms \cite{casares2020quantum,zeguanlcqo,mohammadisiahroudi2023inexact,augustino2023quantum,apers2023quantum}. We refer the reader to the survey paper \cite{abbas2023quantum} for more details on quantum optimization algorithms.
The main caveat in speeding up classical algorithms with QLSAs is the efficient quantum circuit implementation of the oracles assumed for preprocessing the input data and building quantum circuit from the classical system description. Finding such efficient implementations is an area of active research and is crucial for building quantum-classical hybrid algorithms that can provide quantum speed-up.

In this work, we study the use of QLSAs for solving a class of linear systems, whose coefficient matrix and right-hand-side vector can be represented as a linear combination of a few tensor product of 2-by-2 matrices and 2-dimensional vectors, respectively. The problem is a special case of the so-called linear system problem in tensor format (LSP-TF), where matrices/vectors in tensor product are not necessarily 2-dimensional. 
LSP-TF is frequently encountered in discretized PDE problems \cite{grasedyck2013literature}. The problem size grows exponentially as the length of the tensor product chain grows, which makes it difficult to solve using general classic linear system solvers, including Cholesky factorization and Krylov subspaces methods. 
Some classical algorithms have been proposed to solve such LSP-TF \cite{ballani2013projection}. The main idea of these algorithms is to modify the Krylov subspaces methods by taking into account the tensor format of the problem. The resulted Krylov subspaces are constructed approximately, which sacrifices accuracy while keeps the complexity in each iteration polylogarithmic of the problem size. 
Despite the polylogarithmic per-iteration complexity of these algorithms, their total complexity is unknown and is unlikely to be better than that of the original Krylov subspaces for the Krylov subspaces in these algorithms are constructed approximately.
Our main contribution is to show that such linear systems can be efficiently solved using a quantum computer and we provide a full and explicit circuit implementation of the algorithm.

\section{Problem Definition}\label{sec: preliminary}
In this section, we start with notations and then introduce LSP-TF, QLSP, and the Trotterization method. Finally, we summarize our contributions.
%

\subsection{Notation}\label{sec: notation}
Vectors are denoted by lower case letters and matrices are denoted by upper case letters. 
We use $e_i$ to denote the unit basis vector with the $i$th entry being $1$.
We use $I_n$ to denote the identity matrix with dimension $n\times n$, or simply $I$ if the dimension is obvious from the context. Single-qubit Pauli matrices are $\{I,\ X,\ Y, Z\}$, where
\begin{equation*}
    \begin{aligned}
        X = \begin{bmatrix}
            0&1\\1&0
        \end{bmatrix},\
        Y = \begin{bmatrix}
            0&-i\\i&0
        \end{bmatrix},\
        Z = \begin{bmatrix}
            1&0\\0&-1
        \end{bmatrix}.
    \end{aligned}
\end{equation*}
We use
\begin{itemize}
\item $\|\cdot\|_1$ to denote the trace norm of matrices;
\item $\|\cdot\|_2$ to denote the $\ell_2$ norm for vectors and spectral norm for matrices;
\item $\|\cdot\|_F$ to denote the Frobenius norm for matrices.
\end{itemize}
The condition number of a general matrix $M$ is denoted by $\kappa_M$. Let $p$ be a positive integer. We use $[p]$ to represent the set $\{1,2,\dots, p\}$.

We use $\ket{\cdot}$ to represent quantum state, which can be taken as a column vector in this work. Let $\psi = (\psi_1,\dots,\psi_N)$ be a column vector. We use $\ket{\psi}=\sum_{i=1}^N \psi_i\ket{i}/\sqrt{\sum_{i=1}^N |\psi_i|^2}$ to denote the quantum state representation of $\psi$. We also take the convention that
\begin{equation*}
    \ket{0} = \begin{bmatrix}
        1\\0
    \end{bmatrix}
    {\rm and }
    \ket{1} = \begin{bmatrix}
        0\\1
    \end{bmatrix}.
\end{equation*}

\subsection{LSP-TF}\label{sec: LSP-TF}
The general linear system problem (LSP) is to find $x\in \mathbb{R}^N$ such that
\begin{equation}
    Ax=b, \tag{LSP}\label{def: LSP}
\end{equation}
where $A\in \mathbb{R}^{N\times N}$ and $b\in \mathbb{R}^N$. For matrix $A$, its condition number is denoted by $\kappa_A$.
For simplicity, we set $\kappa = \kappa_A$.
The LSP-TF problem is to find $x\in \mathbb{R}^N$ such that
\begin{equation}
Ax=b,\ A = \sum_{i=1}^m  \otimes_{k=1}^n A_{ik},\ b = \sum_{j=1}^d  \otimes_{k=1}^n b_{jk}, \tag{LSP-TF}\label{def: LSP-TF}
\end{equation}
where $A\in \mathbb{R}^{N\times N}$ and $b\in \mathbb{R}^N$.	
This type of problems arises frequently in high-dimensional discretized PDEs \cite{ballani2013projection}. 
In \cite{ballani2013projection}, a combination of classical projection method and low-rank tensor format approximation was proposed to solve \ref{def: LSP-TF}. Their main idea of their method is to use low rank tensors to construct subspaces and maintain the low rank tensor structure using hierarchical Tucker format. Their method obtains per-iteration complexity $\mathcal{O}(mn)$ while assuming $d$ is small.
We refer readers to the survey paper \cite{grasedyck2013literature} on these low tensor rank iterative methods.
%

\subsection{Contributions}
In this work, we apply Algorithm~\ref{algo: qlsa} to solve problem \ref{def: LSP-TF}. We focus on the special case when $A = \sum_{i=1}^m  \otimes_{k=1}^n A_{ik}$ and $b = \sum_{j=1}^d  \otimes_{k=1}^n b_{jk}$ with $A_{ik}\in \mathbb{R}^{2\times 2}$ being Hermitian and $b_{jk}\in \mathbb{R}^2$. Without loss of generality, we make the following assumption.
\begin{assumption}\label{assumption: A b norm 1}
$\|A\|_2\leq 1$, $\|\otimes_{k=1}^n A_{ik}\|_2 \leq 1$ for all $i\in[m]$, and $\|\otimes_{k=1}^n b_{jk}\|_2\leq 1$ for all $j\in[d]$.
\end{assumption}
We also assume $A$ and $b$ are sparse in tensor product strings.
\begin{assumption}
    $m=\mathcal{O}(1)$ and $d=\mathcal{O}(1)$.
\end{assumption}

We propose a Hamiltonian decomposition 
for the Hamiltonian used in Algorithm~\ref{algo: qlsa} when solving \ref{def: LSP-TF} and a detailed circuit implementation for the simulation of the decomposed Hamiltonians. We combine the circuits and the Trotterization method to implement Algorithm~\ref{algo: qlsa}. We show the circuit depth in our implementation of the QLSA is polylogarithmic of the problem dimension. 
Here we provide a simplified version of our main result. The full statement is provided in Theorem~\ref{theorem: main}.
\begin{theorem}\label{theorem: simple main}
Let $0<\epsilon\leq 1/(3n)$, $\epsilon_0 = \epsilon^2/(\kappa\log^2\kappa)$ and $p=\lceil 1- \log_2\epsilon_0 \rceil$. Our implementation of Algorithm~\ref{algo: qlsa} prepares an $\mathcal{O}(\epsilon)$-approximate solution using $\mathcal{T}$  classical arithmetic operations, $\mathcal{T}$ single qubit unitary circuits \textcolor{black}{ and their controlled version with gate depth $\mathcal{O}(1)$}, and $\mathcal{T}$ calls to $p$-qubit quantum multiplier \textcolor{black}{ and their controlled version with circuit depth $\mathcal{O}(\mathcal{T}\mathcal{T}_p)$}, where $\mathcal{T} = \mathcal{O}({\rm poly}(\kappa\log (N)/\epsilon))$ \textcolor{black}{and $\mathcal{T}_p = \mathcal{O}(p^3) $ is the gate depth of $p$-qubit quantum multiplier}. 
\end{theorem}

\begin{remark}
    The final result obtained from our implementation is a quantum state. One needs to apply quantum tomography algorithm to read out any entries.
\end{remark}

The results suggest that the time complexity of our implementation is polylogarithmic in the problem dimension if the time complexity of each gate is $\mathcal{O}(1)$. This is an exponential speed-up compared with classical general linear system algorithms, including Gaussian elimination, Cholesky factorization, and Krylov subspaces methods. When compared with classical algorithms designed for \ref{def: LSP-TF}, for example the algorithm introduced in \cite{ballani2013projection}, our total time complexity is comparable to their per-iteration complexity.

In the remaining of this section, we introduce QLSA and Trotterization method.
%

%

\subsection{QLSA}\label{sec: QLSA}
In this section, we start with the definition of QLSP and give a brief introduction of the QLSA proposed in \cite{subacsi2019quantum}.
\begin{definition}[QLSP]\label{def: QLSP}
Given a Hermitian matrix $A\in \mathbb{R}^{N\times N}$ and $b\in \mathbb{R}^{N}$, the QLSP is to find an approximation of the quantum state
\begin{equation*}
	\ket{x} = \frac{\sum_{i=1}^N x_i \ket{i} }{\sqrt{\sum_{i=1}^N |x_i|^2 }},
\end{equation*}
where $(x_1, \dots, x_N)^\top = A^{-1}b$. 
Quantum state $\ket{\Tilde{x}}$ is called an $\epsilon$-approximate solution to the QLSP if
\begin{equation*}
	\left\| \ket{\Tilde{x}} - \ket{x} \right\|_2 \leq \epsilon.
\end{equation*}
\end{definition}

In \cite{harrow2009quantum}, the authors propose the first QLSA to solve QLSP with polylogarithmic dependence on dimension when the problem is sparse. Later, researchers proposed different QLSAs with the sparsity condition relaxed and better dependence on other parameters including condition number of linear system and solution accuracy \cite{childs2017quantum, chakraborty2018power,subacsi2019quantum,an2022quantum,costa2022optimal,jennings2023efficient}.
In \cite{subacsi2019quantum}, the authors propose a QLSA inspired by adiabatic quantum computing to solve QLSP. They study two types of Hamiltonians and use the randomization method introduced by \cite{boixo2009eigenpath} to design two discretized adiabatic-like quantum algorithms. Algorithm~\ref{algo: qlsa} is the one of the two QLSAs with better complexity. Later in \cite{jennings2023efficient}, an improved version of Algorithm~\ref{algo: qlsa} is proposed with better dependence on the solution accuracy. In this work, we opt to study Algorithm~\ref{algo: qlsa} for its simplicity. For the remainder of this section,we give a brief explanation of Algorithm~\ref{algo: qlsa} and we refer the reader to \cite{subacsi2019quantum} for further details.

Given a linear system problem $Ax=b$, let
\begin{equation*}
    \begin{aligned}
        A(s) = (1-s) Z\otimes I + s X \otimes A
    \end{aligned}
\end{equation*}
and
\begin{equation*}
    \begin{aligned}
        \ket{\bar{b}} &= \frac{\sqrt{2}}{2}(\ket{0}+\ket{1})\otimes\ket{b},\ 
        P_{\bar{b}}^\perp &= I - \ket{\Bar{b}}\bra{\bar{b}},
    \end{aligned}
\end{equation*}
where $s\in [0,1]$ is the evolution parameter.
With these, they introduce the following family of quantum states
\begin{equation*}
	\ket{x(s)} = \frac{\sum_{i=1}^{2N} x_i(s) \ket{i} }{\sqrt{\sum_{i=1}^{2N} |x_i(s)|^2 }},
\end{equation*}
where $(x_1(s), \dots, x_{2N}(s))^\top = A(s)^{-1}\bar{b}$.
They then study the following Hamiltonian \begin{equation}\label{eq: qlsa Hamiltonian}
    \begin{aligned}
        H(s) = \left(\frac{X+iY}{2} \right) \otimes A(s) P_{\Bar{b}}^\perp + \left(\frac{X-iY}{2} \right) \otimes P_{\Bar{b}}^\perp A(s),\ s \in [0,1].
    \end{aligned}
\end{equation}
They prove that, if one starts with quantum state $\ket{0}\otimes \ket{x(0)}$ and lets the quantum system evolve adiabatically, then the final state will be $\ket{0}\otimes \ket{x(1)}$ with sufficiently high probability. 
\textcolor{black}{Notice that, in our framework, the initial state is $\frac{\sqrt{2}}{2}(\ket{0}- \ket{1})\otimes\ket{b}$, where $\ket{b}$ has low tensor rank. This allows us to prepare the initial state efficiently, see the proof of Theorem~\ref{theorem: main} for details.}
They also prove that the relevant spectral gap of the Hamiltonian is bounded from below by the following quantity
\begin{equation*}
	\sqrt{\Delta^\ast(s)} = \sqrt{(1-s)^2 + (s/\kappa_A)^2}.
\end{equation*}
They propose to use the phase randomization method proposed in \cite{boixo2009eigenpath} to turn the continuous adiabatic evolution into piece-wise time-independent evolution. The continuous evolution parameter $s\in [0,1]$, which is time dependent, is discretized into $q$ discretization points $s_j$ for $j\in [q]$.
To determine the discretization of $s$, they parametrize $s$ using
\begin{equation}\label{eq: s parametrization}
s(v) = \frac{e^{v(\sqrt{1+\kappa^2}/\sqrt{2\kappa^2})} + 2\kappa^2 - \kappa^2 e^{-v(\sqrt{1+\kappa^2}/\sqrt{2\kappa^2})}}{2(1+\kappa^2)}.
\end{equation}
Then they discretize $s$ by uniformly discretizing the value of $v$ between $v_a$ and $v_b$, where
\begin{equation*}\label{eq: vavb}
	\begin{aligned}
		v_a &= \frac{\sqrt{2}\kappa}{\sqrt{1+\kappa^2}}\log\left( \kappa\sqrt{1+\kappa^2} - \kappa^2  \right)\\
		v_b &= \frac{\sqrt{2}\kappa}{\sqrt{1+\kappa^2}}\log\left( \sqrt{1+\kappa^2} + 1  \right).
	\end{aligned}
\end{equation*}
At each discretization point $s_j$, the quantum system is evolved using Hamiltonian $H(s_j)$ for time $t_j$, which is sampled from the uniform distribution $t_j \sim [0, 2\pi/\sqrt{\Delta^\ast(s_j)}]$.
The pseudocode of their algorithm is provided in Algorithm~\ref{algo: qlsa}.
\begin{algorithm}
\caption{QLSA for QLSP by \cite{subacsi2019quantum}} 
\label{algo: qlsa}
\begin{algorithmic}[1]
\State Given $A,\ b,\ \kappa,$ and $\epsilon$
\State Compute $v_a$ and $v_b$ by Eq.~\eqref{eq: vavb}
\State Set $q = \Theta\left( \log^2(\kappa)/\epsilon \right)$ and $\delta = (v_b - v_a)/q$
\State For $j=1, \dots, q$, let $v_j = v_a + j\delta,\ s_j = s(v_j)$, and $t_j$ be sampled uniformly from $[0, 2\pi/\sqrt{\Delta^\ast(s_j)}]$
\State Apply $e^{-it_q H(s_q)} \cdots e^{-it_1 H(s_1 )}$ to $\ket{0}\otimes \ket{x(0)}$, discard the ancilla
\end{algorithmic}
\end{algorithm}
As mentioned in \cite{subacsi2019quantum}, the algorithm is inspired by adiabatic quantum computing and thus its run time can be measured by the evolution time needed. The time complexity is $T = \mathcal{O}(\kappa\log(\kappa)/\epsilon)$. 
However, as pointed out by the authors, the Hamiltonian $H(s)$ easily encodes the information of QLSP but not necessarily corresponds to any physical model, which makes the implementation of the algorithm on analog quantum computers difficult.
The authors then introduce a gate-model implementation of Algorithm~\ref{algo: qlsa} and analyze its complexity.
With the oracle access to problem information $A$ and $b$, the authors obtain the query complexity $\tilde{\mathcal{O}}(d\kappa/\epsilon)$, with polylogarithmic factors of $d\kappa/\epsilon$ hidden.

Despite the achievement on complexities, Algorithm~\ref{algo: qlsa} assumes the existence of oracle access to $A$, $b$, and several other non-trivial unitaries. Constructing these oracles from classical input data $(A, b)$ could be more difficult than solving the problem using the QLSA.
In this work, we propose a simpler implementation of Algorithm~\ref{algo: qlsa} and estimate the implementation cost when it is applied to \ref{def: LSP-TF}.

%

\subsection{Trotterization Method}\label{sec: trotter}
In this section, we give a brief introduction to the Trotterization method and its impact on Hamiltonian evolution. The detailed impact of Trotterization method on Algorithm~\ref{algo: qlsa} is analyzed in Section~\ref{sec: trotter cost analysis}.

Consider the time evolution of a time-independent Hamiltonian $M$
, i.e., $e^{-iMt}.$
Since $M$ is a general Hermitian matrix, it is difficult to design the circuit for $e^{-iMt}$.
The main idea of Trotterization methods is that one can consider a decomposition of Hamiltonian $M$, e.g.,
\begin{equation*}
    M = \sum_{\gamma=1}^\Gamma M_{\gamma},
\end{equation*}
where the time evolution of each $M_j$ are are assumed to be easy to compile. Then, for short enough evolution time, one can approximate the Hamiltonian evolution using
\begin{equation*}
    \begin{aligned}
        e^{-it\sum_{\gamma=1}^\Gamma M_{\gamma}} \approx \prod_{\gamma=1}^\Gamma e^{-it M_{\gamma}},
    \end{aligned}
\end{equation*}
which is the first-order Lie-Trotter formula. 
The accuracy of the Trotterization methods is extensively studied in the literature, see e.g.,  \cite{lloyd1996universal,berry2007efficient,childs2019nearly,childs2021theory}. Here we cite Theorem 6 from \cite{childs2021theory} on the accuracy of general Trotterization methods.
We restate it by considering $e^{-iMt}$ instead of $e^{Mt}$ and only for a Hermitian matrix $M$.
\begin{lemma}[Restated Theorem 6 of \cite{childs2021theory}]\label{lemma: childs trotter main}
    Let $M = \sum_{\gamma=1}^\Gamma M_\gamma$ be an Hermitian operator consisting of $\Gamma$ summands with each $M_{\gamma}$ Hermitian and $t\geq 0$. Let $\mathcal{S}(t)$ be a $p$-th order $\Upsilon$-stage product formula as 
    \begin{equation*}
        \begin{aligned}
            \mathcal{S}(t) = \prod_{\upsilon=1}^\Upsilon \prod_{\gamma=1}^\Gamma e^{-it a_{\upsilon, \gamma} M_{\upsilon, \gamma}},
        \end{aligned}
    \end{equation*}
    where $a_{\upsilon, \gamma}$ and $M_{\upsilon, \gamma}$ are to be specified by certain product formula.
    Define $\Tilde{\alpha}_{\rm comm} = \sum_{\gamma_1=1}^\Gamma \cdots \sum_{\gamma_{p+1}=1}^\Gamma \left\| \left[ M_{\gamma_{p+1}}, \cdots \left[M_{\gamma_2}, M_{\gamma_1} \right] \cdots \right] \right\|_2 $. Then, the additive error $\mathcal{E}_{\mathcal{A}}(t)$ and the multiplicative error $\mathcal{E}_{\mathcal{M}}(t)$, defined, respectively, by $\mathcal{S}(t) = e^{-iMt} + \mathcal{E}_{\mathcal{A}}(t)$ and $\mathcal{S}(t) = e^{-iMt}\left( I + \mathcal{E}_{\mathcal{M}}(t) \right)$, can be asymptotically bounded as
    \begin{equation*}
        \begin{aligned}
            \left\| \mathcal{E}_{\mathcal{A}}(t) \right\|_2, \left\| \mathcal{E}_{\mathcal{M}}(t) \right\|_2 = \mathcal{O}\left( \Tilde{\alpha}_{\rm comm} t^{p+1}\right).
        \end{aligned}
    \end{equation*}
\end{lemma}
\begin{remark}\label{remark: alpha communicative factor}
In Lemma \ref{lemma: childs trotter main}, $a_{\upsilon, \gamma}$ and $M_{\upsilon, \gamma}$ are determined in the corresponding Trotterization formula. In this paper, we work with the aforementioned first-order Lie-Trotter formula and thus
\begin{equation*}
	\Upsilon = 1,\ a_{\upsilon, \gamma}=1,\ M_{\upsilon, \gamma} = M_{\gamma},
\end{equation*}
and 
\begin{equation*}
	\begin{aligned}
		\mathcal{S}(t) = \prod_{\gamma=1}^\Gamma e^{-it M_{\gamma}}.
	\end{aligned}
\end{equation*}
The commutator scaling factor becomes
\begin{equation*}
\begin{aligned}
    \Tilde{\alpha}_{\rm comm} &= \sum_{\gamma_1=1}^\Gamma\sum_{\gamma_2=1}^{\Gamma} \left\|[M_{\gamma_1}, M_{\gamma_2}] \right\|_2\\
    &= \sum_{\gamma_1=1}^\Gamma\sum_{\gamma_2=1}^{\Gamma} \left\|M_{\gamma_1}M_{\gamma_2} - M_{\gamma_2}M_{\gamma_1} \right\|_2\\
    &\leq \sum_{\gamma_1=1}^\Gamma\sum_{\gamma_2=1}^{\Gamma} \left\|M_{\gamma_1}M_{\gamma_2}\right\|_2 + \left\| M_{\gamma_2}M_{\gamma_1} \right\|_2\\
    &\leq 2\sum_{\gamma_1=1}^\Gamma\sum_{\gamma_2=1}^{\Gamma} \left\|M_{\gamma_1}\right\|_2 \left\| M_{\gamma_2} \right\|_2.
\end{aligned}
\end{equation*}
\end{remark} 
Lemma \ref{lemma: childs trotter main} applies for general Hamiltonians and thus tighter bounds are possible for well-structured Hamiltonians. In Section \ref{sec: trotter cost analysis}, we discuss the bound of the commutator scaling factor $\Tilde{\alpha}_{\rm comm}$ when the first-order Lie-Trotter formula is applied in the QLSA for LSP-TF. With $\tilde{\alpha}_{\rm comm}$ quantified, one can find the Trotter number. We conclude this section by restating Corollary 7 on Trotter number in \cite{childs2021theory}.
\begin{lemma}[Restated Corollary 7 in \cite{childs2021theory}]\label{lemma: trotter number}
Let $M = \sum_{\gamma=1}^\Gamma M_\gamma$ be a Hermitian operator consisting of $\Gamma$ summands with each $M_{\gamma}$ Hermitian and $t\geq 0$. Let $\mathcal{S}(t)$ be a $p$-th order $\Upsilon$-stage product formula as 
    \begin{equation*}
        \begin{aligned}
            \mathcal{S}(t) = \prod_{\upsilon=1}^\Upsilon \prod_{\gamma=1}^\Gamma e^{-it a_{\upsilon, \gamma} M_{\upsilon, \gamma}}.
        \end{aligned}
    \end{equation*}
Define $\Tilde{\alpha}_{\rm comm} = \sum_{\gamma_1=1}^\Gamma \cdots \sum_{\gamma_{p+1}=1}^\Gamma \left\| \left[ M_{\gamma_{p+1}}, \cdots \left[M_{\gamma_2}, M_{\gamma_1} \right] \cdots \right] \right\|_2 $. Then, we have $\|\mathcal{S}^r(t/r) - e^{-itM}\|_2 = \mathcal{O}(\epsilon)$, provided that 
\begin{equation*}
r = \mathcal{O}\left( \frac{\tilde{\alpha}_{\rm comm}^{1/p} t^{1+1/p}}{\epsilon^{1/p}}  \right).
\end{equation*}
\end{lemma}

The remaining of the paper is organized as follows.  In Section~\ref{sec: circuit design}, we describe our circuit design in the implementation. 
In Section~\ref{sec: circuit analysis}, we analyze the cost to implement the proposed circuits, and finally give the total cost of the implementation.

\section{QLSA Circuit Design for LSP-TF}\label{sec: circuit design}
In this section, we show how to decompose  Hamiltonian \eqref{eq: qlsa Hamiltonian} into two types of structured Hamiltonians when Algorithm~\ref{algo: qlsa} is applied on \ref{def: LSP-TF}. Then we show how to implement the simulation of these two types of Hamiltonians and give their cost estimations.

Following the notation introduced in Section \ref{sec: LSP-TF} and Section \ref{sec: QLSA}, we have
\begin{equation*}
    \begin{aligned}
        A(s) &= (1-s) Z \otimes ( \otimes_{k=1}^n I) + s X \otimes \sum_{i=1}^m  \otimes_{k=1}^n A_{ik}\\
        \ket{\bar{b}} &= \frac{1}{\sqrt{2}}\begin{bmatrix}
            b\\b
        \end{bmatrix} = \frac{1}{\sqrt{2}}\begin{bmatrix}
            \sum_{j=1}^d  \otimes_{k=1}^n b_{jk}\\\sum_{j=1}^d  \otimes_{k=1}^n b_{jk}
        \end{bmatrix}\\
        P_{\bar{b}}^\perp &= I_{2N} - \ket{\bar{b}} \bra{\bar{b}} = I_{2N} - \frac{1}{2} \begin{bmatrix}
            bb^\dag &bb^\dag\\
            bb^\dag &bb^\dag
        \end{bmatrix},
    \end{aligned}
\end{equation*}
where
\begin{equation*}
bb^\dag = \sum_{j_1, j_2=1}^d b_{j_1 \cdot} b_{j_2 \cdot}^\dag = \sum_{j_1, j_2=1}^d\otimes_{k=1}^n b_{j_1k}b_{j_2k}^\dag.
\end{equation*}
Here $b_{j_1 \cdot}$ and $b_{j_2 \cdot}$ represent the corresponding tensor product strings. We also use the similar notation $A_{i\cdot}$.
%
%
With these notations, the Hamiltonian $H(s)$ can be written as
\begin{equation}\label{eq: Hamiltonian decomposition}
    \begin{aligned}
        H(s) &= \left(\frac{X+iY}{2} \right) \otimes A(s) P_{\Bar{b}}^\perp + \left(\frac{X-iY}{2} \right) \otimes P_{\Bar{b}}^\perp A(s)\\
        &= \begin{bmatrix}
            0 & A(s) P_{\bar{b}}^\perp\\
            P_{\bar{b}}^\perp A(s) &0
        \end{bmatrix} \\
             &= H_1(s) + H_2(s) + H_3(s) + H_4(s),
    \end{aligned}
\end{equation}
where
\begin{equation*}
    \begin{aligned}
        H_1(s) &= (1-s) X\otimes Z \otimes ( \otimes_{l=1}^n I)\\
        H_2(s) &= sX \otimes X\otimes A = s\sum_{i=1}^m  X \otimes X\otimes (\otimes_{l=1}^n A_{il})\\
        H_3(s) &= -\frac{1}{2}(1-s)\begin{bmatrix}
        0 & (Z+iY)\otimes bb^\dag\\
        (Z+iY)^\dag \otimes bb^\dag &0
        \end{bmatrix}\\
        &=\frac{1}{2}(1-s) \sum_{j_1, j_2}  \left( \begin{bmatrix}
            0 & Z\otimes b_{j_1 \cdot} b_{j_2 \cdot}^\dag\\
            Z\otimes b_{j_2 \cdot} b_{j_1 \cdot}^\dag & 0
        \end{bmatrix} + \begin{bmatrix}
            0 & iY\otimes b_{j_1 \cdot} b_{j_2 \cdot}^\dag\\
            (iY)^\dag\otimes b_{j_2 \cdot} b_{j_1 \cdot}^\dag & 0
        \end{bmatrix}\right)\\
        H_4(s) &= -\frac{1}{4} s (X + iY) \otimes (I + X) \otimes Abb^\dag -\frac{1}{4} s (X - iY) \otimes (I + X) \otimes bb^\dag A\\
        &= \frac{1}{2}s\sum_{i j_1 j_2}  \left(\begin{bmatrix}
            0& X\otimes \left(A_{i\cdot} b_{j_1 \cdot} b_{j_2 \cdot}^\dag \right)\\
            X\otimes \left( b_{j_2 \cdot} b_{j_1 \cdot}^\dag A_{i\cdot} \right) &0
        \end{bmatrix} + \begin{bmatrix}
            0& I\otimes \left(A_{i\cdot} b_{j_1 \cdot} b_{j_2 \cdot}^\dag \right)\\
            I\otimes \left( b_{j_2 \cdot} b_{j_1 \cdot}^\dag A_{i\cdot} \right) &0
        \end{bmatrix}\right).
    \end{aligned}
\end{equation*}
The summands in the four Hamiltonians can be categorized into the two types of Hamiltonian defined in the following.
Let us define the two types of Hamiltonian
\begin{equation*}
    \begin{aligned}
        \text{Type-1: } H^1 &= P_1\otimes P_2 \otimes C_1 \otimes \cdots \otimes C_n\\
        \text{Type-2: }H^2 &= \begin{bmatrix}
            0 & P_0 \otimes D_1 \otimes \cdots \otimes D_n\\
            P_0^\dag \otimes D_1^\dag \otimes \cdots \otimes D_n^\dag &0
        \end{bmatrix},
    \end{aligned}
\end{equation*}
where $P_j \in \{I, X, Y, Z\}$, $C_i$ and $D_i$ are both $\mathbb{C}^{2\times2}$, and $C_i$ is Hermitian.
It is obvious that the summands of $H_1(s)$ and $H_2(s)$ are Type 1 Hamiltonians; the summands of $H_3(s)$ and $H_4(s)$ are Type 2 Hamiltonians. 
One can also verify that the total amount of such summands in $H(s)$ is a polynomial of $m$ and $d$.
%
\begin{lemma}\label{lemma: qlsa Hamiltonian decomposition count}
$H(s)$ can be represented as the summation of $(m+1)$ Type-1 Hamiltonians and $(2d^2 + 2md^2)$ Type-2 Hamiltonians.
\end{lemma}
\begin{proof}
We see that $H_1(s)$ is Type-1; $H_2(s)$ is summation of $m$ Type-1 Hamiltonians; $H_3(s)$ is summation of $2d^2$ Type-2 Hamiltonians; and $H_4(s)$ is summation of $2md^2$ Type-2 Hamiltonians.
\end{proof}
\begin{remark}
The fact that $H(s)$ can be decomposed into $\mathcal{O}(md^2)$ Type-1 and Type-2 Hamiltonians allows us to apply Trotterization Methods to implement Algorithm~\ref{algo: qlsa} because it will only contribute polylogarithmic overhead to complexities.
\end{remark}

We apply the Trotterization method to the evolution of Hamiltonian $H(s)$ and approximate it by the product of time evolution of all the summands, which are Type 1 and Type 2 Hamiltonians. In the next section,  we introduce our circuit design for the two types of circuits and discuss their properties. In Section \ref{sec: circuit analysis}, we analyze the circuit approximation accuracy needed for the Troterization Method and the total circuit depth of the whole implementation.

\subsection{Circuit for Type-1 Hamiltonian}
A Type-1 Hamiltonian is a tensor product of Pauli and Hermitian matrices. $H_1(s)$ is a special case of Type-1 Hamiltonian since it is a scaled Pauli matrix. For Pauli matrices, their time evolution is discussed in \cite{nielsen2010quantum} Section 4.7.3. The idea is to apply single qubit operations on each qubit to turn any $X$ or $Y$ to $Z$; then apply the matrix exponential for the new Hamiltonian with only $I$ and $Z$; finally undo the operations to turn certain $Z$ back to $X$ or $Y$. We refer the readers to \cite{nielsen2010quantum} for details about this case.

For general Type-1 Hamiltonians, the following lemma is helpful for our circuit design.
\begin{lemma}\label{lemma: unitary sandwich evolution}
	For any $t\geq 0$ and any Hermitian matrix $M$ with decomposition $M = U_P \Sigma_M U_P^\dag$, where $U_P$ is a unitary matrix, the following identity holds
	\begin{equation*}
		e^{-itM} = U_P e^{-it \Sigma_M} U_P^\dag.
	\end{equation*}
\end{lemma}
\begin{proof}
Using Taylor expansion, it is obvious that
\begin{align*}
	e^{-itM} &= I + (-it)M + \frac{(-it)^2}{2!}M^2 + \cdots \\
	&=  U_P\left(I + (-it)\Sigma_M + \frac{(-it)^2}{2!}\Sigma_M^2 + \cdots  \right) U_P^\dag \\
	&= U_P e^{-it \Sigma_M} U_P^\dag.
\end{align*}
\end{proof}

\begin{remark}
For Type-1 Hamiltonian $H^1$, it can be decomposed as
\begin{equation*}
\begin{aligned}
	H^1 &=  U_{H^1} \Sigma_{H^1} U_{H^1}^\dag
\end{aligned}
\end{equation*}
with
\begin{equation*}
\begin{aligned}
	U_{H^1} &= U_{P_1} \otimes U_{P_2} \otimes U_{C_1} \otimes \cdots \otimes U_{C_n}\\
	\Sigma_{H^1} &= \Sigma_{P_1}\otimes \Sigma_{P_2} \otimes \Sigma_{C_1} \otimes \cdots \otimes \Sigma_{C_n},
\end{aligned}
\end{equation*}
where $U_{P_i} \Sigma_{P_i} U_{P_i}^\dag$ is the eigenvalue decomposition of $P_i$ and $U_{C_j}\Sigma_{C_j}U_{C_j}^\dag$ is the singular value decomposition of $C_j$.
We choose the order of the singular values such that $\Sigma_{H^1}$ can be represented as
\begin{equation} \label{eq: type1 sv}
 \begin{aligned}
 	\Sigma_{H^1} &= \left\| \Sigma_{H^1} \right\|_2 \Sigma_{P_1}\otimes \Sigma_{P_2} \otimes\begin{bmatrix}
 	1&0\\0& \sigma_{H^1,1}
 	\end{bmatrix}\otimes \cdots 
 	\otimes \begin{bmatrix}
 	1&0\\0& \sigma_{H^1,n}
 	\end{bmatrix},
 \end{aligned}
\end{equation}
where $\sigma_{H^1, i}\in [0, 1]$ for all $i\in \{1, \dots, n\}.$
Then the time evolution can be written as
\begin{equation*}
e^{-itH^1} = U_{H^1} e^{-it\Sigma_{H^1}} U_{H^1}^\dag.
\end{equation*}
\end{remark}

For $U_{H^1}$, each matrix in the tensor product is a unitary matrix in $\mathbb{C}^{2\times2}$, thus we can implement the circuit for them efficiently. 
For $e^{-it\Sigma_{H^1}}$, we introduce an algorithm to implement the circuit. For the simplicity and the ease to read, we start with a generic version, see Algorithm \ref{algo: generic implement diagonal tensor}, and give explanations afterwards.
%
%
%
%
\begin{algorithm}[H] 
\caption{Generic Implementation of $e^{-it\Sigma}$} 
\label{algo: generic implement diagonal tensor}
\begin{algorithmic}[1]
\State Given $t\geq 0$ and $\Sigma$ in the format of Eq.~\eqref{eq: type1 sv}
\State Represent diagonal element $\sigma_{i}$ in binary format,
\State Encode the binary representation of $\sigma_{i}$ as quantum state $\ket{\sigma_{i}}$,
\State Apply quantum multiplier to $\ket{\sigma_{i}}$ controlled by the input quantum states,
\State Apply a series of Phase Gate controlled by the result of quantum multiplier.
%
\end{algorithmic}
\end{algorithm}
Step 2-4 in Algorithm \ref{algo: generic implement diagonal tensor} are inspired by \cite{ruiz2017quantum}. We encourage the readers to read \cite{ruiz2017quantum} to learn their quantum adder and quantum multiplier. Their quantum adder and quantum multiplier are designed for integer numbers.
In the followings, we give a thorough explanation of each step. For Step 4,  we start with a brief introduction of the quantum adder and multiplier. Then, we show that the quantum adder and multiplier also work for rational numbers in the specific form described in Lemma~\ref{lemma: adder} and explain how they can be used in the implementation. 

\subsubsection{Algorithm \ref{algo: generic implement diagonal tensor} Step 2 \& 3}
By the definition of $\sigma_i$ in Eq.~\eqref{eq: type1 sv}, we have $\sigma_i\in[0,1]$. Here we temporarily assume $\sigma_i$ can be represented using $p$ binary digits
\begin{equation*}
\begin{aligned}
\sigma_i = \sum_{j=1}^p \nu_j^i 2^{-j},
\end{aligned}
\end{equation*}
where $\nu_j^i$ are binary numbers.
This assumption is relaxed in Section~\ref{sec: approximate data}, where inexact binary representation is discussed.
Then, we use $\ket{\nu_j^i} = \ket{0}$ to represent $\nu_j^i = 0$ and $\ket{\nu_j^i} = \ket{1}$ to represent $\nu_j^i = 1$, which gives us a quantum representation of $\sigma_i$ as
\begin{equation*}
\ket{\sigma_i} = \ket{\nu_j^1}\ket{\nu_j^2}\cdots\ket{\nu_j^p}.
\end{equation*}
Such binary presentation of numbers is used to construct quantum adder and multiplier in \cite{ruiz2017quantum}.

\subsubsection{Algorithm \ref{algo: generic implement diagonal tensor} Step 4 \& 5}
In \cite{ruiz2017quantum}, the authors introduce a non-modular quantum Fourier transform (QFT) based quantum adder following the QFT adder introduced in \cite{draper2000addition}. Then they introduce a quantum multiplier based on their quantum adder.
The quantum adder and the quantum multiplier are designed for positive integers originally and easily accommodate signed integers by adding an extra qubit for the sign.
Staring from Lemma~\ref{lemma: adder}, we show that the quantum adder and the quantum multiplier also work for rational numbers between $-1$ and $1$ with mild modification of the original version. 
With the QFT-based quantum adder and quantum multiplier, we are ready to explain Step 4 of Algorithm~\ref{algo: generic implement diagonal tensor}.

Notice that $\Sigma$ is a diagonal matrix, it is known that $e^{-it\Sigma}$ is also a diagonal matrix with the $j$th diagonal entry being $e^{-it\Sigma_j}$, where $\Sigma_j$ is the $j$th entry of $\Sigma$.
When we apply the circuit $e^{-it\Sigma}$ on quantum state $\ket{j}$, we expect $e^{-it\Sigma_j} \ket{j}$ being the resulting state, i.e., $e^{-it\Sigma} \ket{j} = e^{-it\Sigma_j} \ket{j}$. Once we have a binary representation of $\Sigma_j$ encoded as a quantum state, we apply a gate phase circuit to turn $\ket{j}$ into $e^{-it\Sigma_j} \ket{j}$.
Specifically, for $\ket{\Sigma_j} = \ket{\nu_j^1}\cdots\ket{\nu_j^p}$, if  $\nu_j^k=1$, then we apply the phase shift $e^{-it 2^{-k}}$ to $\ket{j}$. This is the controlled Phase gate in Step 5.
What is left is the computation of the binary representation of $\Sigma_j$. Notice that $\Sigma$ is a tensor product of $(n+2)$ $2$-by-$2$  diagonal matrices with all the diagonal entries between $-1$ and $1$, we can get $\Sigma_j$ by multiplying $n+2$ numbers. 
Specifically, for $\ket{j} = \ket{j_1}\ket{j_2}\cdots\ket{j_{n+2}}$ with $j_k$ being binary numbers, if  $j_k = 0$, then we take the first diagonal entry of the $k$th $2$-by-$2$ diagonal matrix into the multiplication; otherwise, we take the second diagonal entry of the matrix. This is the controlled quantum multiplier in Step 4. The quantum circuit we describe maps $\ket{j}$ to $e^{-it\Sigma_j} \ket{j}$ and thus implements the circuit $e^{-it\Sigma}$ because of the linearity of quantum circuit.

In the rest of this section, we prove our aforementioned claim that the quantum adder and quantum multiplier can be generalized to certain non-integer numbers in Lemma~\ref{lemma: adder} to \ref{lemma: multiplier gate}. Finally, we summarize our circuit for Type-1 Hamiltonian afterwards in Lemma~\ref{lemma: Type 1 exact}.
\begin{lemma}\label{lemma: adder}
Let $K$ be an integer. Let $\nu^1 = \nu^1_1 2^{K-1} + \cdots + \nu^1_p 2^{K-p}$ and $\nu^2 = \nu^2_1 2^{K-1} + \cdots + \nu^2_p 2^{K-p}$,
where $\nu^1_i$ and $\nu^2_j$ are binary numbers for all $i\in[p]$ and $j\in[p]$.
Let $\ket{\nu^1} = \ket{\nu^1_1}\otimes\cdots\otimes\ket{\nu^1_p}$ and $\ket{\nu^2} = \ket{\nu^2_1}\otimes\cdots\otimes\ket{\nu^2_p}$. If $C$ is a circuit such that, when $K=p$,
\begin{equation}\label{eq: adder circuit}
C \ket{0}\otimes\ket{\nu^1}\otimes\ket{\nu^2} = \ket{\nu^3}\otimes\ket{\nu^2},
\end{equation}
where $\ket{\nu^3} = \ket{\nu^3_1}\otimes \cdots\otimes\ket{\nu^3_{p+1}}$, $\nu^3_i\in \{0,1\}$ for all $i\in[p+1]$, and 
\begin{equation}\label{eq: adder}
\nu^1+\nu^2 = \nu^3_1 2^{K} + \cdots + \nu^3_{p+1} 2^{K-p},
\end{equation}
then, Eq.~\eqref{eq: adder circuit} holds for any integer $K$.
\end{lemma}
\proof{Proof.}
When $K=p$, the circuit $C$ is the quantum adder introduced in \cite{ruiz2017quantum}, which applies to integer numbers. Since Eq.~\eqref{eq: adder} holds for $K=p$, we have
\begin{equation*}
2^K\left( \sum_{i=1}^p \nu^1_i 2^{-i} + \sum_{i=1}^p \nu^2_i 2^{-i} - \sum_{i=1}^{p+1} \nu^3_i 2^{1-i} \right) =0,
\end{equation*}
which holds for any integer $K$.
So the circuit $C$ is an adder for any number pair $\nu^1$ and $\nu^2$.
\endproof
We have shown the quantum adder works for positive rational numbers in the specific form described in Lemma~\ref{lemma: adder}. In fact, it also works for corresponding signed rational numbers after mild modification following the argument in \cite{ruiz2017quantum}. This modification only adds one extra qubit for the sign and does not change the asymptotic number of the gates needed to implement the adder. In this work, we only care about the addition of numbers between $-1$ and $1$, which is the case when $K=0$.
For the adder to accomplish this job, we summarize its properties in the following lemma. Readers are referred  to Section 4 in \cite{ruiz2017quantum} for the proof.
\begin{lemma}\label{lemma: adder gate}
Let $p$ be an integer. Let $\nu^1$ and $\nu^2$ be signed numbers with $|\nu^1| = \nu^1_1 2^{-1} + \cdots + \nu^1_p 2^{-p}$ and $|\nu^2| = \nu^2_1 2^{-1} + \cdots + \nu^2_p 2^{-p}$,
where $\nu^1_i$ and $\nu^2_j$ are binary numbers for all $i\in[p]$ and $j\in[p]$.
The quantum adder introduced in \cite{ruiz2017quantum} computes $\nu^1+\nu^2$ using $\mathcal{O}(p^2)$ single qubit gates and their controlled versions \textcolor{black}{with gate depth $\mathcal{O}(p^2)$}. Specifically, the adder uses one $(p+2)$-qubit QFT, one $(p+2)$-qubit IQFT, and $(p+1)(p+2)/2$ controlled single-qubit gates.
\end{lemma}
Similar argument can be made for the quantum multiplier. We summarize them in the following two lemmas.
\begin{lemma}\label{lemma: multiplier}
Let $K$ be an integer. Let $\nu^1 = \nu^1_1 2^{K-1} + \cdots + \nu^1_p 2^{K-p}$ and $\nu^2 = \nu^2_1 2^{K-1} + \cdots + \nu^2_p 2^{K-p}$,
where $\nu^1_i$ and $\nu^2_j$ are binary numbers for all $i\in[p]$ and $j\in[p]$.
Let $\ket{\nu^1} = \ket{\nu^1_1}\otimes\cdots\otimes\ket{\nu^1_p}$ and $\ket{\nu^2} = \ket{\nu^2_1}\otimes\cdots\otimes\ket{\nu^2_p}$. If $C$ is a circuit such that, when $K=p$,
\begin{equation}\label{eq: multiplier circuit}
C \ket{0}_{2p}\otimes\ket{\nu^1}\otimes\ket{\nu^2} = \ket{\nu^3}\otimes\ket{\nu^1}\otimes\ket{\nu^2},
\end{equation}
where $\ket{\nu^3} = \ket{\nu^3_1}\otimes \cdots\otimes\ket{\nu^3_{2p}}$, $\nu^3_i\in \{0,1\}$ for all $i\in [2p]$, and 
\begin{equation}\label{eq: multiplier}
\nu^1 \nu^2 = \nu^3_1 2^{2K-1} + \cdots + \nu^3_{2p} 2^{2K-2p},
\end{equation}
then, Eq.~\eqref{eq: multiplier circuit} holds for any integer $K$.
\end{lemma}
\proof{Proof.}
When $K=p$, the circuit $C$ is the original quantum multiplier introduced in \cite{ruiz2017quantum}, which applies to integer numbers. Since Eq.~\eqref{eq: multiplier} holds for $K=p$, we have
\begin{equation*}
2^{2K}\left( \sum_{i=1}^p \nu^1_i 2^{-i} \sum_{i=1}^p \nu_i^2 2^{-i} - \sum_{i=1}^{2p} \nu_i^3 2^{-i}  \right)=0,
\end{equation*}
which holds for any integer $K$. So the circuit $C$ is an multiplier for any number pair $\nu^1$ and $\nu^2$.
\endproof
In \cite{ruiz2017quantum}, the authors also discuss how to multiply signed integer numbers by adding qubits for signs. We refer the readers to Section 7 of \cite{ruiz2017quantum} and summarize the main result in the following lemma.
\begin{lemma}\label{lemma: multiplier gate}
Let $\nu^1$ and $\nu^2$ be two signed numbers with $|\nu^1| = \nu^1_1 2^{-1} + \cdots + \nu^1_p 2^{-p}$ and $|\nu^2| = \nu^2_1 2^{-1} + \cdots + \nu^2_p 2^{-p}$,
where $\nu^1_i$ and $\nu^2_j$ are binary numbers for all $i\in[p]$ and $j\in[p]$.
The quantum multiplier introduced in \cite{ruiz2017quantum} computes $\nu^1\nu^2$ using $\mathcal{O}(p^3)$ single-qubit gates and their controlled version \textcolor{black}{with gate depth $\mathcal{O}(p^3)$}. Specifically, the multiplier uses one $(2p+2)$-qubit QFT, one $(2p+2)$-qubit IQFT, and $p$ controlled adders.
\end{lemma}
For both the quantum adder and quantum multiplier, when their input numbers are represented by $p$-qubit quantum states, we call the adder and multiplier by $p$-qubit adder and $p$-qubit multiplier, respectively.
The following lemma discusses the complexity of Algorithm~\ref{algo: generic implement diagonal tensor}.
\begin{lemma}\label{lemma: svd evolution}
If $\sigma_{H^1,i} = \sum_{j=1}^p \nu_j^i 2^{-j}$ with $\nu_j^i$'s being binary for all $i\in [n]$, there exists a quantum circuit that prepares $e^{-it\Sigma_{H^1}}$ as described in Algorithm~\ref{algo: generic implement diagonal tensor}, for any $t\geq 0$, using at most $\mathcal{O}(n)$ calls to the $p$-qubit quantum multiplier \textcolor{black}{with gate depth $\mathcal{O}(np^3)$}.
\end{lemma}
\proof{Proof.}
Step 4 in Algorithm~\ref{algo: generic implement diagonal tensor} computes the multiplication of at most $n+2$ $p$-qubit numbers and thus needs at most $n+2$ calls to the $p$-qubit multiplier.
 \endproof
The following lemma summarizes the cost to perform time evolution of Type-1 Hamiltonian $H^1$.
\begin{lemma}\label{lemma: Type 1 exact}
For Type-1 Hamiltonian $H^1$ with decomposition
\begin{equation*}
\begin{aligned}
	H^1 &=  U_{H^1} \Sigma_{H^1} U_{H^1}^\dag
\end{aligned}
\end{equation*}
with
\begin{equation*}
\begin{aligned}
	U_{H^1} &= U_{P_1} \otimes U_{P_2} \otimes U_{C_1} \otimes \cdots \otimes U_{C_n}\\
	\Sigma_{H^1} &= \Sigma_{P_1}\otimes \Sigma_{P_2} \otimes \Sigma_{C_1} \otimes \cdots \otimes \Sigma_{C_n},
\end{aligned}
\end{equation*}
where $U_{P_i} \Sigma_{P_i} U_{P_i}^\dag$ is the eigenvalue decomposition of $P_i$ and $U_{C_j}\Sigma_{C_j}U_{C_j}^\dag$ is the singular value decomposition of $C_j$. We can rewrite $\Sigma_{H^1}$ as
\begin{equation*}
 \begin{aligned}
 	\Sigma_{H^1} &= \left\| H^1 \right\|_2 \Sigma_{P_1}\otimes \Sigma_{P_2} \otimes\begin{bmatrix}
 	1&0\\0& \sigma_{H^1,1}
 	\end{bmatrix}\otimes \cdots 
 	\otimes \begin{bmatrix}
 	1&0\\0& \sigma_{H^1,n}
 	\end{bmatrix},
 \end{aligned}
\end{equation*}
where $\sigma_{H^1, i}\in [0, 1]$ for all $i\in \{1, \dots, n\}.$
If $\sigma_{H^1,i} = \sum_{j=1}^p \nu_j^i 2^{-j}$  with $\nu_j^i$'s being binary for all $i\in [n]$, then there exists a quantum circuit that prepares $e^{-itH^1}$ for any $t\geq0$, using $\mathcal{O}(n)$ classical arithmetic operations, $\mathcal{O}(n)$ single qubit unitary circuits \textcolor{black}{with gate depth $\mathcal{O}(1)$}, and $\mathcal{O}(n)$ calls to the $p$-qubit quantum multiplier \textcolor{black}{with gate depth $\mathcal{O}(np^3)$}.  
\end{lemma}
\proof{Proof.}
We start by computing the singular value decomposition of $H^1$. Considering $H^1$ is a tensor product of $n+2$ two-by-two matrices, the singular value decomposition of $H^1$ is the tensor product of the singular value decomposition of the $2$-by-$2$ matrices. For each $2$-by-$2$ matrix, the number of classical arithmetic operations needed for computing singular value decomposition is a constant. So in total $\mathcal{O}(n)$ classical arithmetic operations are needed. 
Then, we need to construct the circuit for unitary operation $U_{H^1}$, which is the tensor product of $n+2$ two-by-two matrices. So we need $\mathcal{O}(n)$ single qubit unitary circuits.
Finally, we need to construct the circuit for $e^{-it\Sigma_{H^1}}$. According to Lemma~\ref{lemma: svd evolution}, we need $\mathcal{O}(n)$ calls to the $p$-qubit quantum multiplier.
 \endproof
%
%


\subsection{Circuit for Type-2 Hamiltonian}
Unlike Type-1 Hamiltonians, in general, Type-2 Hamiltonians are not in the tensor format,
which brings extra challenges for implementation. In this section, we introduce a method to avoid the issue and show that circuit for Type-2 Hamiltonian can be implemented in a similar way as Type-1 Hamiltonian.

First, we classically compute the singular value decomposition of $D_1 \otimes \cdots \otimes D_n$.
This step takes at most $\mathcal{O}\left( n \right)$ classical arithmetic operations. Denote the singular value decomposition of $P_0 \otimes D_1 \otimes \cdots \otimes D_n$ by
\begin{equation*}
\begin{aligned}
	P_0 \otimes D_1 \otimes \cdots \otimes D_n &=  U_{H^2} \Sigma_{H^2} V_{H^2}^\dag,
\end{aligned}
\end{equation*}
where
\begin{equation*}
\begin{aligned}
	U_{H^2} &= P_0 \otimes U_{D_1} \otimes \cdots \otimes U_{D_n}\\
	V_{H^2} &= I \otimes V_{D_1} \otimes \cdots \otimes V_{D_n}\\
	\Sigma_{H^2} &= \left\| H^2 \right\|_2 I\otimes \begin{bmatrix}
 	1&0\\0& \sigma_{H^2,1}
 	\end{bmatrix}\otimes \begin{bmatrix}
 	1&0\\0& \sigma_{H^2,2}
 	\end{bmatrix} \otimes \cdots 
 	\otimes \begin{bmatrix}
 	1&0\\0& \sigma_{H^2,n}
 	\end{bmatrix}
\end{aligned}
\end{equation*}
with $\sigma_{H^2,i}\in[0,1]$ for all $i\in[n]$.
Then, we have
\begin{equation*}
H^2=\begin{bmatrix}
            U_{H^2} & 0\\
            0& V_{H^2}
        \end{bmatrix} \times \begin{bmatrix}
            0 & \Sigma_{H^2}\\
            \Sigma_{H^2} &0
        \end{bmatrix} \times \begin{bmatrix}
            U_{H^2}^\dag &0\\
            0 & V_{H^2}^\dag
        \end{bmatrix}.
\end{equation*}
Note that the decomposition is not a singular value decomposition but it is useful for the implementation because the outer matrices are unitaries and the inner matrix is in tensor format, where
Lemma~\ref{lemma: unitary sandwich evolution} applies.
The lemma indicates that we need to implement the circuit of the outer unitary matrix and the circuit for the time evolution of the inner matrix.

For the outer matrix, it is unitary but it is not in the tensor format. We can further decompose it as follows
\begin{equation*}
    \begin{bmatrix}
            U_{H^2} & 0\\
            0& V_{H^2}
        \end{bmatrix} = L_{H^2} R_{H^2},
\end{equation*}
where 
\begin{equation*}
\begin{aligned}
L_{H^2} &= \begin{bmatrix}
            U_{H^2} & 0\\
            0& I_{2^{n+1}}
        \end{bmatrix},\ R_{H^2} = \begin{bmatrix}
            I_{2^{n+1}} & 0\\
            0& V_{H^2}
        \end{bmatrix}.
\end{aligned}
\end{equation*}
$L_{H^2}$ and $R_{H^2}$ are still unitary and represent the controlled version of $U_{H^2}$ and $V_{H^2}$, respectively.
Since $U_{H^2}$ and $V_{H^2}$ are tensor product of $2$-by-$2$ unitaries, we can further decompose $L_{H^2}$ and $R_{H^2}$ into
\begin{equation*}
\begin{aligned}
L_{H^2} &= \begin{bmatrix}
P_0 \otimes I \otimes \cdots \otimes I &0\\0&I_{2^{n+1}}
\end{bmatrix}
\begin{bmatrix}
I \otimes U_{D_1} \otimes \cdots \otimes I &0\\0&I_{2^{n+1}}
\end{bmatrix}
\cdots
\begin{bmatrix}
I \otimes I \otimes \cdots \otimes U_{D_n} &0\\0&I_{2^{n+1}}
\end{bmatrix}
\\
R_{H^2} &= \begin{bmatrix}
I_{2^{n+1}}&0\\0 &I \otimes I \otimes \cdots \otimes I
\end{bmatrix}
\begin{bmatrix}
I_{2^{n+1}}&0\\0&I \otimes V_{D_1} \otimes \cdots \otimes I 
\end{bmatrix}
\cdots
\begin{bmatrix}
I_{2^{n+1}}&0\\0&I \otimes I \otimes \cdots \otimes V_{D_n} 
\end{bmatrix}.
\end{aligned}
\end{equation*}
Each matrix in the decomposition of $L_{H^2}$ and $R_{H^2}$ is a single qubit unitary controlled by $n+1$ qubits and can be implemented efficiently.

According to Lemma~\ref{lemma: unitary sandwich evolution}, we still need to implement $e^{-it X\otimes \Sigma_{H^2}}$. However, $X\otimes \Sigma_{H^2}$ is not a diagonal matrix and thus cannot be implemented directly using Algorithm~\ref{algo: generic implement diagonal tensor}. To void this, we consider the eigenvalue decomposition of $X$ and have
\begin{equation}\label{eq: X Sig decompose}
\begin{aligned}
X \otimes\Sigma_{H^2} &= (U_{X} \Sigma_{X} U_{X}^\dag)\otimes\Sigma_{H^2}\\
&= (U_X \otimes I_{2^{n+1}})(\Sigma_X \otimes \Sigma_{H^2})(U_X \otimes I_{2^{n+1}})^\dag.
\end{aligned}
\end{equation}
Now we apply Lemma~\ref{lemma: unitary sandwich evolution} again since $U_X\otimes I_{2^{n+1}}$ is a unitary. Also, since $\Sigma_X = \begin{bmatrix}
1&0\\0&-1
\end{bmatrix}$, the matrix $\Sigma_X \otimes \Sigma_{H^2}$ is a tensor product of $2$-by-$2$ diagonal matrices, where Algorithm~\ref{algo: generic implement diagonal tensor} applies. We summarize the analysis above in the following lemma.

\begin{lemma}\label{lemma: Type 2 exact}
For Type-2 Hamiltonian $H^2$, it can be decomposed into
\begin{equation*}
H^2=\begin{bmatrix}
            U_{H^2} & 0\\
            0& V_{H^2}
        \end{bmatrix} \times \begin{bmatrix}
            0 & \Sigma_{H^2}\\
            \Sigma_{H^2} &0
        \end{bmatrix} \times \begin{bmatrix}
            U_{H^2}^\dag &0\\
            0 & V_{H^2}^\dag
        \end{bmatrix}
\end{equation*}
with
\begin{align}
	U_{H^2} &= P_0 \otimes U_{D_1} \otimes \cdots \otimes U_{D_n} \notag\\
	V_{H^2} &= I \otimes V_{D_1} \otimes \cdots \otimes V_{D_n} \notag\\
	\Sigma_{H^2} &= \left\| H^2 \right\|_2 I\otimes \begin{bmatrix}
 	1&0\\0& \sigma_{H^2,1}
 	\end{bmatrix}\otimes \begin{bmatrix}
 	1&0\\0& \sigma_{H^2,2}
 	\end{bmatrix} \otimes \cdots 
 	\otimes \begin{bmatrix}
 	1&0\\0& \sigma_{H^2,n}
 	\end{bmatrix}, \label{eq: type2 sv}
\end{align}
where $U_{H^2}\Sigma_{H^2}V_{H^2}^\dag$ is the singular value decomposition of $P_0 \otimes D_1 \otimes \cdots \otimes D_n$ and $\sigma_{H^2,i}\in[0,1]$ for all $i\in[n]$.
If $\sigma_{H^2,i} = \sum_{j=1}^p \nu_j^i 2^{-j}$  with $\nu_j^i$'s being binary for all $i\in [n]$, then there exists a quantum circuit that prepares $e^{-itH^2}$ for any $t\geq0$, using $\mathcal{O}(n)$ classical arithmetic operations, $\mathcal{O}(n)$ single qubit unitary circuits, and $\mathcal{O}(n)$ calls to $p$-qubit quantum multiplier.  
\end{lemma}
\proof{Proof.}
According to the analysis in Section~\ref{sec: circuit design}, the decomposition exists. According to Lemma~\ref{lemma: unitary sandwich evolution}, we need to implement unitary circuit $\begin{bmatrix}
            U_{H^2} & 0\\
            0& V_{H^2}
        \end{bmatrix}$ and unitary circuit $exp\left\{-it \begin{bmatrix}
            0 & \Sigma_{H^2}\\
            \Sigma_{H^2} &0
        \end{bmatrix} \right\}$. The first one can be implemented using $\mathcal{O}(n)$ controlled single qubit unitary circuits. The second one can be decomposed using Eq.~\eqref{eq: X Sig decompose}. After decomposition, one can implement $U_X\otimes I_{2^{n+1}}$ using one controlled single qubit unitary circuit. As for the time evolution circuit $exp\left\{ -it \Sigma_X\otimes \Sigma_{H^2} \right\} $, the Hamiltonian $\Sigma_X\otimes \Sigma_{H^2}$ is in the same form of $\Sigma_{H^1}$, so we can use Algorithm~\ref{algo: generic implement diagonal tensor} to implement the time evolution circuit, which comes with at most $\mathcal{O}(n)$ calls to $p$-qubit quantum multiplier according to Lemma~\ref{lemma: svd evolution}.
 \endproof

In this section, we showed how to implement the time evolution for Type-1 and Type-2 Hamiltonian and give cost estimation for both of them. However, we assume that all the numbers $\sigma_{H^j, i}$ can be represented using $p$ qubits, which does not hold for the general case. We discuss this issue in Section \ref{sec: circuit analysis}.

\section{Circuit Cost Analysis}\label{sec: circuit analysis}
In this section, we analyze the total cost needed to use the circuits proposed in the previous section to implement Algorithm~\ref{algo: qlsa}. Our analysis are twofold: we first estimate the cost needed to implement the circuits when we do not have exact $p$-qubit binary representation of data, and then we estimate the Trotterization steps needed for Algorithm~\ref{algo: qlsa} to converge.

\subsection{Cost Analysis for Approximate Data}\label{sec: approximate data}
For the analysis in Section~\ref{sec: circuit design}, we assume $\sigma_{H^1,i}$ in Eq.~\eqref{eq: type1 sv} and $\sigma_{H^2, j}$ in Eq.~\eqref{eq: type2 sv} can be represented by $p$-qubit quantum states. This assumption does not hold for the general case.
In fact, we do not need an exact binary representation of these numbers. Here we introduce the definition of $\epsilon$-approximate binary representation of a number $\nu\in[0,1]$ and then discuss the behavior of quantum adder and multiplier on approximate binary representations. Finally, we discuss the relationship between the inexactness of data and the inexactness of the circuit constructed by Algorithm~\ref{algo: generic implement diagonal tensor}.
\begin{definition}\label{def: approximate binary}
Let $\nu\in[0,1]$, $\epsilon>0$, and $p \in \mathbb{Z}_+$. Let $\nu_i \in \{0, 1\}$ for all $i\in \{1,2,\dots,p\}$. We call 
\begin{equation*}
\tilde{\nu} = \nu_1 2^{-1} + \nu_{2} 2^{-2} + \dots + \nu_ 2^{-p}
\end{equation*}
a $p^{\rm th}$-order $\epsilon$-approximation of $\nu$, if
\begin{equation*}
|\nu - \tilde{\nu}| \leq \epsilon.
\end{equation*}
The corresponding binary representation is called a $p^{\rm th}$-order $\epsilon$-approximate binary representation of $\nu$.
\end{definition}
In this work, we use $p$ qubits to store a $p^{\rm th}$-order binary representation, so we also call the approximation a $p$-qubit approximation and the binary representation a  $p$-qubit binary representation.
\begin{lemma}\label{lemma: approximate binary}
Let $\nu\in[0,1]$, $\epsilon>0$, and $p_0=\lceil 1-\log_2(\epsilon) \rceil$. There exists $\nu_i \in \{0, 1\}$ for all $i\in \{1,2,\dots,p_0\}$ such that
\begin{equation*}
\tilde{\nu} = \nu_1 2^{-1} + \nu_{2} 2^{-2} + \dots + \nu_{p_0} 2^{-p_0}
\end{equation*}
is a $p_0$-qubit $\epsilon$-approximation $\nu$. 
\end{lemma}
\proof{Proof.}
Consider the exact binary representation of $\nu$:
\begin{equation*}
\begin{aligned}
\nu = \sum_{i=1}^{+\infty} \nu_i 2^{-i}
= \tilde{\nu} + \sum_{i=p_0+1}^{+\infty} \nu_i 2^{-i}.
\end{aligned}
\end{equation*}
It is obvious that
\begin{equation*}
|\nu - \tilde{\nu}| = \sum_{i=K+1}^{+\infty} \nu_i 2^{-i} \leq \sum_{i=p_0+1}^{+\infty} 2^{-i} \leq 2^{-p_0} \leq \epsilon.  
\end{equation*}
\endproof
Lemma \ref{lemma: approximate binary} indicates that we can obtain an $\epsilon$-approximate binary representation of a number using at most $\lceil 1-\log_2(\epsilon) \rceil$ qubits.
In the next lemma, we discuss the quantum multiplication of two $p$-qubit approximate numbers.
\begin{lemma}\label{lemma: 2 inexact multiplication}
Let $0<\epsilon<\epsilon^j<1$  for $j\in[2]$, $p=\lceil 1-\log_2(\epsilon) \rceil$, and $0\leq \nu^1, \nu^2\leq 1$. Let $ \tilde{\nu}^j = \sum_{i=1}^p \nu_i^j 2^{-i}$ be the $p$-qubit $\epsilon^j$-approximation of $\nu^j$ for $j\in[2]$. Let $\nu^3 = \sum_{i=1}^{2p} \nu^3_i 2^{-i}$ be $\tilde{\nu}^1\tilde{\nu}^2$ computed by the quantum multiplier introduced in \cite{ruiz2017quantum}. Then, $\nu^3$ is a $2p$-qubit $(\epsilon^1+\epsilon^2+\epsilon^1\epsilon^2)$-approximation of $\nu^1\nu^2$. Let $\tilde{\nu}^3 = \sum_{i=1}^{p} \nu^3_i 2^{-i}$ be an approximation of $\nu^3$. Then, $\tilde{\nu}^3$ is a $p$-qubit $(\epsilon +\epsilon^1+\epsilon^2+\epsilon^1\epsilon^2)$-approximation of $\nu^1\nu^2$.
\end{lemma}
\proof{Proof.}
According to Lemma~\ref{lemma: approximate binary}, both $\tilde{\nu}^1$ and $\tilde{\nu}^2$ exist. Since $\tilde{\nu}^1$ is a $p$-qubit $\epsilon$-approximation of $\nu^1$, we have
$\tilde{\nu}^1 = \nu^1 + \delta^1$
with $|\delta^1|\leq \epsilon^1$. Similarly for $\nu^2$, we have $\tilde{\nu}^2 = \nu^2 + \delta^2$ with $|\delta^2|\leq \epsilon^2$. 
The quantum multiplier gives an exact $2p$-qubit binary representation of $\nu^3 = \tilde{\nu}^1\tilde{\nu}^2$. By their definitions, we have
\begin{equation*}
\begin{aligned}
\left| \nu^3 - \nu^1\nu^2  \right| &= \left| (\nu^1 + \delta^1)(\nu^2 + \delta^2) - \nu^1\nu^2  \right|\\
&=\left| \nu^1\delta^2 + \nu^2\delta^1 + \delta^1\delta^2   \right|\\
&\leq \epsilon^1+\epsilon^2+\epsilon^1\epsilon^2.
\end{aligned}
\end{equation*}
Since $\nu^3$ is represented by $2p$ qubits, one can discard the last $p$ less-important qubits, leaving a $p$-qubit approximation of $\nu^3$, i.e., $\tilde{\nu}^3$. As a result, we have
\begin{equation*}
\begin{aligned}
\left| \tilde{\nu}^3 - \nu^1\nu^2  \right| &= \left| \tilde{\nu}^3 - \nu^3 + \nu^3 - \nu^1\nu^2  \right|\\
&\leq \epsilon + \left| \nu^3 - \nu^1\nu^2  \right|,
\end{aligned}
\end{equation*}
where the last inequality holds because
\begin{equation*}
\left|\nu^3 - \tilde{\nu}^3 \right|= \sum_{i=p+1}^{2p} \nu^3_i 2^{-i} \leq \sum_{i=p+1}^{2p} 2^{-i} \leq 2^{-p} \leq \epsilon.  
\end{equation*}
\endproof
As indicated in the previous lemma, when multiplying two $p$-qubit numbers, one can either keep the $2p$-qubit result or approximate it by a $p$-qubit approximation. However, when considering the multiplication of multiple $p$-qubit numbers,  since we stick with the quantum multiplier introduced in \cite{ruiz2017quantum},  we have to consider the following issue. 
For example, if we add an extra $p$-qubit number $\nu^4$ and we want the result of $\nu^1\nu^2\nu^4$. Then, after we get $\nu^3$ from the quantum multiplier, we have to decide whether to use $\tilde{\nu}^3\nu^4$ as an approximate result, or to turn $\nu^4$ into a $2p$ qubit number and obtain a $4p$-qubit number as the result.
If we apply the second strategy, then number of qubits needed would grow exponentially as the amount of numbers grow. 
Instead, if we apply the first strategy, for every extra number, we need to use $p$ more qubits, which grows linearly. We call the first strategy by truncation strategy.
In this work, we stick with the quantum multiplier introduced in \cite{ruiz2017quantum} and choose the truncation strategy described above.

In the next lemma, we discuss the quantum multiplication of multiple $p$-qubit approximate numbers.
\begin{lemma}\label{lemma: approximate multiplier}
Let $K$ be a positive integer larger than 1. Let $0<\epsilon<1$ such that $K\epsilon<1/3$. Let $p=\lceil 1-\log_2(\epsilon) \rceil$.
Let $0\leq \nu^j\leq 1$ for $j\in[K]$
and $ \tilde{\nu}^j = \sum_{i=1}^p \nu_i^j 2^{-i}$ be the $p$-qubit $\epsilon$-approximation of $\nu^j$ for $j\in[K]$. Let $\nu^\ast$ be the result computed by the quantum multiplier introduced in \cite{ruiz2017quantum} following the truncation strategy. Then, $\nu^\ast$ is a $p$-qubit $\epsilon_K$-approximation of $\prod_{j=1}^K \nu^j$, where
\begin{equation*}
\begin{aligned}
\epsilon_K \leq 3K\epsilon.
\end{aligned}
\end{equation*}
\end{lemma}
\proof{Proof.}
We prove the lemma using induction. Define function $f(k) = 2\sum_{i=1}^k (k\epsilon)^i$, for $k=2,\dots.$ 
It is obvious that $f(k)\leq 2k\epsilon/(1-k\epsilon)$ when $0<k\epsilon<1$.
When $k=2$, according to Lemma~\ref{lemma: 2 inexact multiplication}, we have
\begin{equation*}
\epsilon_2 = 3\epsilon+\epsilon^2 \leq 2(2\epsilon + 4\epsilon^2) = f(2).
\end{equation*}
Let us assume that when $k=K$, $\nu^+$ is a $p$-qubit $\epsilon_K$-approximation of $\prod_{j=1}^K \nu^j$ with $\epsilon_K\leq f(K)$. Then, when $k=K+1$, according to Lemma~\ref{lemma: 2 inexact multiplication}, it follows that
\begin{equation*}
\begin{aligned}
\epsilon_{K+1} &= \epsilon+\epsilon_{K} + \epsilon + \epsilon_{K}\epsilon\\
&\leq 2\epsilon + f(K) + f(K)\epsilon\\
&= 2\epsilon + \left( 2K\epsilon + 2\sum_{i=2}^K (K\epsilon)^i \right)  + \left(2 \sum_{i=1}^{K-1}(K\epsilon)^i \epsilon + 2 K^K \epsilon^{K+1} \right) \\
&= 2(K+1)\epsilon + 2\sum_{i=2}^K(K^i + K^{i-1})\epsilon^i + 2K^K\epsilon^{K+1}\\
&\leq 2(K+1)\epsilon +2 \sum_{i=2}^K (K+1)^i\epsilon^i + 2(K+1)^{K+1}\epsilon^{K+1}\\
&= f(K+1).
\end{aligned}
\end{equation*}
So we conclude $\epsilon_K\leq f(K) \leq 2K\epsilon/(1-K\epsilon)\leq 3K\epsilon$.
 \endproof

\begin{lemma}\label{lemma: inexact data and circuit}
Let $0<\epsilon<\frac{1}{3n}$ and $p=\lceil 1-\log_2(\epsilon) \rceil$. Let $\tilde{\sigma}_{H^1, i}$ be a $p$-qubit $\epsilon$-approximation of $\sigma_{H^1, i}$ for all $i\in[n]$. When we use $\tilde{\sigma}_{H^1, i}$ in Algorithm~\ref{algo: qlsa} and obtain an inexact circuit $e^{-it\tilde{\Sigma}_{H^1}}$, the following bound holds
\begin{equation*}
\left\| e^{-it\tilde{\Sigma}_{H^1}} - e^{-it\Sigma_{H^1}}  \right\|_2 \leq t\epsilon.
\end{equation*}
\end{lemma}
\proof{Proof.}
Both $e^{-it\tilde{\Sigma}_{H^1}}$ and $e^{-it\Sigma_{H^1}}$ are diagonal matrices, it follows that
\begin{equation*}
\begin{aligned}
\left\| e^{-it\tilde{\Sigma}_{H^1}} - e^{-it\Sigma_{H^1}}  \right\|_2 &\leq \max_{i\in[n]} \left| e^{-it\tilde{\sigma}_{H^1,i}} - e^{-it\sigma_{H^1,i}}  \right|\\
&= \max_{i\in[n]} \left| (e^{-it(\tilde{\sigma}_{H^1,i}-\sigma_{H^1,i}) } - 1)e^{-it\sigma_{H^1,i}}  \right|\\
&= \max_{i\in[n]} \left| e^{-it(\tilde{\sigma}_{H^1,i}-\sigma_{H^1,i}) } - 1 \right|\\
&= \max_{i\in[n]} \sqrt{\left( e^{-it(\tilde{\sigma}_{H^1,i}-\sigma_{H^1,i}) } - 1 \right) \left( e^{-it(\tilde{\sigma}_{H^1,i}-\sigma_{H^1,i}) } - 1 \right)^\dag}\\
&= \max_{i\in[n]} \sqrt{2-2\cos\left(t \tilde{\sigma}_{H^1,i}- t\sigma_{H^1,i}\right)}\\
&\leq \max_{i\in[n]}t\left|\tilde{\sigma}_{H^1,i}- \sigma_{H^1,i}\right|\\
&\leq t\epsilon.  
\end{aligned}
\end{equation*}
\endproof

\subsection{Trotterization Cost Analysis}\label{sec: trotter cost analysis}
In Section~\ref{sec: circuit design}, we introduce how to implement the time evolution circuits of two types of Hamiltonian. When we use those circuits in the implementation of Algorithm~\ref{algo: qlsa}, we decompose the Hamiltonian \eqref{eq: qlsa Hamiltonian} into some structured Hamiltonians as shown in Eq.\eqref{eq: Hamiltonian decomposition}.
Then we use the time evolution circuits of these structured Hamiltonians to approximate the time evolution circuits used in Algorithm~\ref{algo: qlsa}, which is the Trotterization method. As mentioned in Section~\ref{sec: trotter}, %
Trotterization method introduces error into circuits and makes circuits inexact. Lemma~\ref{lemma: childs trotter main} provides a bound for the Trotterization error in terms of communicative factor $\tilde{\alpha}_{\rm comm}$. In the next lemma, we provide a bound for $\tilde{\alpha}_{\rm comm}$ when the Hamiltonian decomposition in Eq.~\eqref{eq: Hamiltonian decomposition} is used by the specific Trotterization method -- the first-order Lie-Trotter formula. 
\begin{lemma}\label{lemma: alpha communicative final}$\tilde{\alpha}_{\rm comm} \leq 2(1+m+2d^2 +2md^2)^2 = \mathcal{O}(1)$.
\end{lemma}
\proof{Proof.}
As discussed in Remark~\ref{remark: alpha communicative factor} and Lemma~\ref{lemma: qlsa Hamiltonian decomposition count}, we need to find the spectral norm of the $(m+1)$ Type-1 Hamiltonians and the $(2d^2 + 2md)$ Type-2 Hamiltonians.
According to Eq.~\eqref{eq: Hamiltonian decomposition}, $H(s)$ is decomposed into $4$ Hamiltonians, each Hamiltonian is further decomposed into Type-1 and Type-2 Hamiltonians.
For the single Type-1 Hamiltonian in $H_1(s)$, we have
\begin{equation*}
\|X\otimes Z\otimes (\otimes_{l=1}^n I)\|_2 = 1.
\end{equation*}
For the $m$ Type-1 Hamiltonians in $H_2(s)$, we have
\begin{equation*}
\|X\otimes X\otimes (\otimes_{l=1}^n A_{il})\|_2\leq 1
\end{equation*}
according to Assumption~\ref{assumption: A b norm 1}.

For the $2d^2$ Type-2 Hamiltonians in $H_3(s)$, denote the singular value decomposition of $Z\otimes b_{j_1 \cdot} b_{j_2 \cdot}^\dag$ by $U_{Zb} \Sigma_{Zb} V_{Zb}^\dag$, it follows that
\begin{equation*}
\begin{aligned}
\left\| \begin{bmatrix}
            0 & Z\otimes b_{j_1 \cdot} b_{j_2 \cdot}^\dag\\
            Z\otimes b_{j_2 \cdot} b_{j_1 \cdot}^\dag & 0
        \end{bmatrix} \right\|_2 &= \left\| \begin{bmatrix}
        U_{Zb} &0\\0&V_{Zb}
\end{bmatrix}
\begin{bmatrix}
0 & \Sigma_{Zb}\\ \Sigma_{Zb}^\dag &0
\end{bmatrix}
\begin{bmatrix}
U_{Zb}^\dag&0\\0&V_{Zb}^\dag
\end{bmatrix}           \right\|_2 \\
&= \left\| \begin{bmatrix}
0 & \Sigma_{Zb}\\ \Sigma_{Zb}^\dag &0
\end{bmatrix} \right\|_2\\
&= \left\| \Sigma_{Zb} \right\|_2\\
&\leq \left\| b_{j_1\cdot} b_{j_2 \cdot}^\dag  \right\|_2\\
&= \left\| b_{j_1\cdot} \right\|_2 \left\|b_{j_2 \cdot}^\dag  \right\|_2\\
&\leq 1,
\end{aligned}
\end{equation*}
where the last inequality holds because of Assumption~\ref{assumption: A b norm 1}.
Similar for the remaining Type-2 Hamiltonians in $H_3(s)$.

For the $2md^2$ Type-2 Hamiltonians in $H_4(s)$, notice that
\begin{equation*}
\left\| A_{i\cdot}b_{j_1\cdot} b_{j_2 \cdot}^\dag \right\|_2 \leq \left\| A_{i\cdot}\right\|_2 \left\|b_{j_1\cdot} b_{j_2 \cdot}^\dag \right\|_2\leq 1,
\end{equation*}
it follows that all the Type-2 Hamiltonians in $H_4(s)$ have spectral norm bounded by $1$.
Combine these bounds with Lemma~\ref{lemma: qlsa Hamiltonian decomposition count}, we have
\begin{equation*}
\tilde{\alpha}_{\rm comm} \leq 2(1+m+2d^2 +2md^2)^2.  
\end{equation*}
\endproof

\subsection{Total Cost Analysis}
Both inexact data and Trotterization method introduce errors into the unitary circuits. In this section, we discuss the relationship between the inexactness of circuits and the accuracy of final state and then estimate the total cost needed to implement Algorithm~\ref{algo: qlsa}.
\begin{lemma}\label{lemma: inexact unitary product error}
Let $K$ be a positive integer and $\epsilon>0$. Let $U_i$ for $i\in[K]$ be unitary circuits and $\tilde{U}_i$ for $i\in[K]$ be the corresponding inexact circuits such that $\Tilde{U}_i = U_i + \mathcal{E}_i$ for all $i =[K]$.
Then, the following inequality holds
\begin{equation*}
\left\| \prod_{i=1}^K \tilde{U}_i - \prod_{i=1}^K U_i \right\|_2 \leq \sum_{i=1}^K \|\mathcal{E}_i\|_2.
\end{equation*}
\end{lemma}
\proof{Proof.}
Although the circuits $\tilde{U}_i$ are inexact, they are still unitary operators. So the following identity holds
\begin{equation*}
\prod_{i=1}^k\left(U_i+\mathcal{E}_i\right) -  \prod_{i=1}^k U_i = \left(  \prod_{i=1}^{k-1}\left(U_i+\mathcal{E}_i\right) -  \prod_{i=1}^{k} U_i (U_k + \mathcal{E}_k)^\dag \right) (U_k+ \mathcal{E}_k)
\end{equation*}
for all $k\in[K]$ because $U_k+\mathcal{E}_k$ is unitary.
From the above identity, we have
\begin{equation*}
\begin{aligned}
\left\| \prod_{i=1}^K \tilde{U}_i - \prod_{i=1}^K U_i \right\|_2 &= \left\| \prod_{i=1}^K\left(U_i+\mathcal{E}_i\right) -  \prod_{i=1}^K U_i \right\|_2\\
&= \left\| \left(  \prod_{i=1}^{K-1}\left(U_i+\mathcal{E}_i\right) -  \prod_{i=1}^{K-1} U_i (U_K + \mathcal{E}_K)^\dag \right) (U_K+ \mathcal{E}_K)   \right\|_2\\
&= \left\|  \prod_{i=1}^{K-1}\left(U_i+\mathcal{E}_i\right) -  \prod_{i=1}^{K} U_i (U_K + \mathcal{E}_K)^\dag   \right\|_2\\
&= \left\|  \prod_{i=1}^{K-1}\left(U_i+\mathcal{E}_i\right) -  \prod_{i=1}^{K-1} U_i - \prod_{i=1}^{K} U_i\mathcal{E}_K^\dag   \right\|_2\\
&\leq \left\|  \prod_{i=1}^{K-1}\left(U_i+\mathcal{E}_i\right) -  \prod_{i=1}^{K-1} U_i \right\|_2 + \left\| \prod_{i=1}^{K} U_i\mathcal{E}_K^\dag   \right\|_2\\
&\leq \left\|  \prod_{i=1}^{K-1}\left(U_i+\mathcal{E}_i\right) -  \prod_{i=1}^{K-1} U_i \right\|_2 + \left\| \mathcal{E}_K   \right\|_2\\
&\quad \vdots\\
&\leq \sum_{i=1}^K \|\mathcal{E}_i\|_2,
\end{aligned}
\end{equation*}
where the third equality and the second inequality hold because unitary operators do not change spectral norm and the last inequality holds because of recursion. 
 \endproof

The following lemma comes from \cite{jennings2023efficient} discussing the trace distance between exact and inexact matrices when applied to the same density operator.
\begin{lemma}[Lemma 9 in \cite{jennings2023efficient}]\label{lemma: trace distance inequality}
Let $U_i$ for $i\in[2]$ be square matrices and $\tilde{U}_i$ for $i\in[2]$ be the corresponding approximate matrices such that 
$\tilde{U}_i = U_i + \mathcal{E}_i$ with $\|\mathcal{E}_i\|_2 \leq \epsilon$ for $i\in[2]$. Then for every density operator $\rho$, the following inequality holds
\begin{equation*}
\left\| \tilde{U}_1\rho \tilde{U}_2 - U_1\rho U_2  \right\|_1 \leq \epsilon(\|U_1\|_2 + \|U_2||_2) + \epsilon^2.
\end{equation*}
\end{lemma}

Now we are ready to provide the main theorem of this work.
\begin{theorem}\label{theorem: main}
Let $0<\epsilon\leq 1/(3n)$, $\epsilon_0 = \epsilon^2/(\kappa\log^2\kappa)$ and $p=\lceil 1- \log_2\epsilon_0 \rceil$. When using Trotter number $r= \mathcal{O}(m^2d^4\kappa^2/\epsilon)$ for Trotterization method and $p$-qubit $\epsilon_0$-approximation of $\sigma_i$ in Algorithm~\ref{algo: generic implement diagonal tensor}, our implementation of Algorithm~\ref{algo: qlsa} prepares an $\mathcal{O}(\epsilon)$-approximate solution using $\mathcal{T}$  classical arithmetic operations, $\mathcal{T}$ single qubit unitary circuits \textcolor{black}{ and their controlled version with gate depth $\mathcal{O}(1)$}, and $\mathcal{T}$ calls to $p$-qubit quantum multiplier \textcolor{black}{ and their controlled version with gate depth $\mathcal{O}(\mathcal{T}p^3)$}, where $\mathcal{T} =  \mathcal{O}({ \kappa^2log(N)\log^2(\kappa)/\epsilon^2})$. 
\end{theorem}
\proof{Proof.}
In Algorithm~\ref{algo: qlsa}, denote the initial state by $\ket{\psi_0}$ and the ideal final state by $\ket{\psi^\ast}$.
\textcolor{black}{We first discuss the cost to implement Algorithm~\ref{algo: qlsa} and finally demonstrate that the cost to prepare the initial is negligible.}
Notice that evolution time $t=(t_1,\dots,t_q)$ is a $q$-dimensional random variable, we denote the circuits by $U_j(t_j) = e^{-it_j H(s^j)}$ and $U(t) = \prod_{j=1}^q U_j(t_j)$. According to \cite{subacsi2019quantum}, when $U(t)$ is implemented exactly, then the expected trace distance between the actual and ideal final states is bounded by
\begin{equation*}
\mathbb{E}_t  \left\| U(t)\ket{\psi_0}\bra{\psi_0}U(t)^\dag - \ket{\psi^\ast}\bra{\psi^\ast} \right\|_1 \leq \epsilon.
\end{equation*}
In our implementation, each unitary circuit $U_j(t_j)$ is approximated by a unitary circuit $\tilde{U}_j(t_j)$ satisfying $\tilde{U}(t) = U(t) + \mathcal{E}(t)$.
It follows that
\begin{equation*}
\begin{aligned}
&\mathbb{E}_t  \left\| \tilde{U}(t)\ket{\psi_0}\bra{\psi_0}\tilde{U}(t)^\dag - \ket{\psi^\ast}\bra{\psi^\ast} \right\|_1\\
=& \mathbb{E}_t  \left\| \tilde{U}(t)\ket{\psi_0}\bra{\psi_0}\tilde{U}(t)^\dag - U(t)\ket{\psi_0}\bra{\psi_0}U(t)^\dag +U(t)\ket{\psi_0}\bra{\psi_0}U(t)^\dag - \ket{\psi^\ast}\bra{\psi^\ast} \right\|_1\\
\leq & \mathbb{E}_t  \left(\left\| \tilde{U}(t)\ket{\psi_0}\bra{\psi_0}\tilde{U}(t)^\dag - U(t)\ket{\psi_0}\bra{\psi_0}U(t)^\dag\right\|_1  + \left\| U(t)\ket{\psi_0}\bra{\psi_0}U(t)^\dag - \ket{\psi^\ast}\bra{\psi^\ast} \right\|_1\right)\\
\leq & \mathbb{E}_t \left\| \tilde{U}(t)\ket{\psi_0}\bra{\psi_0}\tilde{U}(t)^\dag - U(t)\ket{\psi_0}\bra{\psi_0}U(t)^\dag\right\|_1  + \epsilon\\
\leq & \mathbb{E}_t \left( 2\|\mathcal{E}(t)\|_2 + \|\mathcal{E}(t)\|_2^2 \right)+ \epsilon,
\end{aligned}
\end{equation*}
where the last inequality holds because of Lemma~\ref{lemma: trace distance inequality}.

As we know, $U(t)$ is the product of $q$ circuits $U_j(t_j)$. Each $U_j(t_j)$ is approximated by $\tilde{U}_j(t_j)$, whose inexactness comes from Trotterization method and inexact data. We use $\bar{U}_j(t_j)$ to denote the circuit that approximates $U_j(t_j)$ using the same Trotterization method with accurate data.
Let $r_j$ be the Trotter number for $U_j(t_j)$. Then
\begin{equation*}
\begin{aligned}
\left\| \tilde{U}_j(t_j) - U_j(t_j)  \right\|_2 &= \left\| \tilde{U}_j(t_j) - \bar{U}_j(t_j) + \bar{U}_j(t_j) - U_j(t_j)  \right\|_2\\
&\leq \left\| \tilde{U}_j(t_j) - \bar{U}_j(t_j)\right\|_2 + \left\| \bar{U}_j(t_j) - U_j(t_j)  \right\|_2
\end{aligned}
\end{equation*}
The first term measures the error introduced by inexact data. There are $r_j$ circuits in both $\tilde{U}_j(t_j)$ and $\bar{U}_j(t_j)$ since they are both Trotterized with Trotter number $r_j$, i.e.,
\begin{equation*}
\tilde{U}_j(t_j) = \prod_{l=1}^{r_j} \tilde{U}_j(t_j/r_j),\ \bar{U}_j(t_j) = \prod_{l=1}^{r_j} \bar{U}_j(t_j/{r_j}).
\end{equation*}
According to Lemma~\ref{lemma: inexact unitary product error} and Lemma~\ref{lemma: inexact data and circuit}, we have
\begin{equation*}
\begin{aligned}
\left\| \tilde{U}_j(t_j) - \bar{U}_j(t_j)\right\|_2 &\leq r_j \left\| \tilde{U}_j(t_j/r_j) - \bar{U}_j(t_j/r_j)\right\|_2\\
&\leq t_j \epsilon_0.
\end{aligned}
\end{equation*}
By the definition of $t_j$ in Algorithm~\ref{algo: qlsa}, $t_j\leq 2\pi\kappa$. It follows that
\begin{equation*}
\left\| \tilde{U}_j(t_j) - \bar{U}_j(t_j)\right\|_2 \leq 2\pi\epsilon^2/\log^2\kappa.
\end{equation*}
As for the second term, it measures the error introduced by Trotterization method. According to Lemma~\ref{lemma: trotter number} and Lemma~\ref{lemma: alpha communicative final}, we have $\|\bar{U}_j(t_j) - U_j(t_j)\|_2 = \mathcal{O}(\epsilon^2/\log^2\kappa)$, provided that Trotter number is bounded by $\mathcal{O}\left( {m^2d^4 t_j^2 }/{\epsilon^2}   \right)$. In this work, we choose Trotter number be the upper bound of this bound, which guarantees $\|\bar{U}_j(t_j) - U_j(t_j)\|_2 = \mathcal{O}(\epsilon^2/\log^2\kappa)$.
Put these together, we have
\begin{equation*}
\left\| \tilde{U}_j(t_j) - U_j(t_j)  \right\|_2 = \mathcal{O}(\epsilon^2/\log^2\kappa).
\end{equation*}
Following Lemma~\ref{lemma: inexact unitary product error}, we have
\begin{equation*}
\begin{aligned}
\mathbb{E}_t \left\| \mathcal{E}(t)  \right\|_2 &= \mathbb{E}_t \left\| \tilde{U}(t) - U(t)    \right\|_2\\
&\leq \mathbb{E}_t \sum_{j=1}^q \left\| \tilde{U}_j(t_j) - U_j(t_j)  \right\|_2 \\
&=\mathcal{O}(\epsilon),
\end{aligned}
\end{equation*}
where the last inequality holds by the definition of $q$ in Algorithm~\ref{algo: qlsa}. Finally, we have
\begin{equation*}
\mathbb{E}_t  \left\| \tilde{U}(t)\ket{\psi_0}\bra{\psi_0}\tilde{U}(t)^\dag - \ket{\psi^\ast}\bra{\psi^\ast} \right\|_1 \leq \mathbb{E}_t \left(2\|\mathcal{E}(t)\|_2 + \|\mathcal{E}(t)\|_2^2 \right) + \epsilon = \mathcal{O}(\epsilon).
\end{equation*}

In this implementation of Algorithm~\ref{algo: qlsa}, there are $q$ circuits $e^{-it_jH(s_j)}$ and each of them is approximated using $r$ Trottered circuits. Each Trottered circuits consists of $\mathcal{O}(md^2)$ time evolution circuits of Type-1 or Type-2 Hamiltonians. So the total cost of our implementation is $\mathcal{T}$  classical arithmetic operations, $\mathcal{T}$ single qubit unitary circuits that can be performed in parallel, and $\mathcal{T}$ calls to $p$-qubit quantum multiplier, where $\mathcal{T} = \mathcal{O}(qrmnd^2) = \mathcal{O}(nm^3 \kappa^2 d^6 \log^2\kappa/\epsilon^2)$. 

Finally, we discuss the cost to prepare the initial state.
\textcolor{black}{We first discuss the cost to prepare the initial state.
The initial state $\ket{\psi_0} =\frac{\sqrt{2}}{2}(\ket{0}-\ket{1}) \otimes \ket{b} $ has $\mathcal{O}(1)$ low tensor rank and thus can be implemented efficiently. We provide two methods to prepare the initial state and show that in both cases we can ignore the cost to prepare the initial state. The first method starts from the fact that, for each $b_j = \otimes_{k=1}^n b_{jk}$, we have
\begin{equation*}
    \ket{b_j} = \otimes_{k=1}^n\left(\frac{1}{\|b_{jk}\|_2}\begin{bmatrix}
        b_{jk} & \tilde b_{jk}
    \end{bmatrix} \ket{0}\right), 
\end{equation*}
where $\tilde b_{jk}$ is the complex conjugate of $-iYb_{jk}$ and $\frac{1}{\|b_{jk}\|_2}\begin{bmatrix}
        b_{jk} & \tilde b_{jk}
    \end{bmatrix}\in \mathbb{C}^{2\times 2}$ is unitary.
When $d=1$, we only need $\mathcal{O}(n)$ classical arithmetic operations and $\mathcal{O}(n)$ single qubit gates to prepare the initial state, which is negligible compared with the cost to implement Algorithm~\ref{algo: generic implement diagonal tensor}.
When $d=\mathcal{O}(1)$, instead of solving $Ax=b$ directly, we can solve $d$ linear systems $Ax = \otimes_{k=1}^n b_{jk}$ for $j\in[d]$. The total complexity would be $d$ times the complexity to implement Algorithm~\ref{algo: generic implement diagonal tensor} when $d=1$. This does not change the asymptotic total complexity since $d=\mathcal{O}(1)$.
One can also prepare the initial state directly without solving $d$ linear systems. Similar as  earlier, it takes $\mathcal{O}(dn)$ classical arithmetic operations and $\mathcal{O}(dn)$ single qubit gates \textcolor{black}{ with gate depth $\mathcal{O}(d)$ } to prepare $d$ unitaries $U_{b_j}$ such that $U_{b_j}\ket{0} = \ket{b_j}$ for $j\in[d]$. We claim the following circuit can be used to prepare the initial state. 
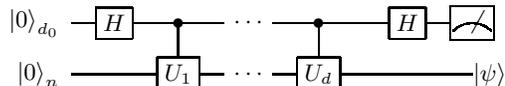
\[
\Qcircuit @C=1em @R=.7em {
    \lstick{\ket{0}_{d_0}} & \gate{H} & \ctrl{1} & \qw & \cdots & & \ctrl{1} & \qw & \gate{H} & \meter  \\
    \lstick{\ket{0}_n} & \qw & \gate{U_1} & \qw & \cdots & & \gate{U_{d}} & \qw & \qw & \qw & \lstick{\ket{\psi}}
}
\]
The first register has $d_0 = \lceil \log_2(d) \rceil$ qubits and represents states $\ket{1}$ to $\ket{d_0}$. For the controlled unitaries, if the first register is in state $\ket{j}$ with $j\in[d]$, then $U_{b_j}$ is applied on the second register. Finally, when state $\ket{1}$ is measured in the first register, the second register is in state $\ket{b}$. It takes $\mathcal{O}(d)$ tries to get $\ket{b}$, with which, the initial state can be easily prepared. In this case, the total cost to prepare the initial state is still negligible.
}

\endproof

Notice that, when $d>1$, one can solve $d$ linear systems $Ax=  \otimes_{k=1}^n b_{jk}$ for $j\in[d]$ and use the $d$ solutions to construct the solution of the original linear systems. So we have the following corollary.
\begin{corollary}
    Let $0<\epsilon\leq 1/(3n)$, $\epsilon_0 = \epsilon^2/(\kappa\log^2\kappa)$ and $p=\lceil 1- \log_2\epsilon_0 \rceil$. When using Trotter number $r= \mathcal{O}(m^2d^4\kappa^2/\epsilon)$ for Trotterization method and $p$-qubit $\epsilon_0$-approximation of $\sigma_i$ in Algorithm~\ref{algo: generic implement diagonal tensor}, one can use our implementation of Algorithm~\ref{algo: qlsa} to prepare an $\mathcal{O}(\epsilon)$-approximate solution using $\mathcal{T}$  classical arithmetic operations, $\mathcal{T}$ single qubit unitary circuits \textcolor{black}{ and their controlled version with gate depth $\mathcal{O}(1)$}, and $\mathcal{T}$ calls to $p$-qubit quantum multiplier \textcolor{black}{ and their controlled version with gate depth $\mathcal{O}(Tp^3)$}, where $\mathcal{T} = \mathcal{O}(nm^3 \kappa^2 d \log^2\kappa/\epsilon^2)$.
\end{corollary}



\section{Conclusion}
In this work, we study using an adiabatic inspired quantum linear system algorithm to solve linear system problem in tensor format. We focus on a class of linear systems whose matrices consist of a linear combination of tensor products of Hermitian $2$-by-$2$ matrices and linear system vector consists of a linear combinations of tensor products of $2$-dimensional vectors. 
We explicitly describe all the quantum circuits components used in the implementation of the QLSA. The implementation only uses single qubit unitary circuits, $p$-qubit quantum multipliers, and their controlled versions. The number of classical arithmetic operations and the number of these gates are polynomial of $n$, $m$, and $d$, and thus polylogarithmic in the dimension of the linear system. 
Considering the execution time for each of these gates is $\mathcal{O}(1)$, the total time complexity of the implementation is polylogarithmic in the problem dimension, which is better than any classical algorithm for general linear system problems. 
When compared with the classical algorithm designed for linear system in tensor format, our total complexity is comparable to the single step complexity of the classical algorithm proposed by \cite{ballani2013projection} in terms of the problem dimension. 

Possible interesting generalizations of our method would be to extend these results to higher dimensional $A_{ik}$ and $b_{jk}$ and to cases where $A_{ik}$ and/or $b_{jk}$ are in different sub-spaces. It is also worth continuing investigating how to obtain better dependence on condition number and accuracy for the structured problems.

\section*{Acknowledgement}
{The authors thank Faisal Alam and Muqing Zheng for helpful discussions. 
{This work was supported by Defense Advanced Research Projects Agency as part of the project W911NF2010022: {\em The Quantum
Computing Revolution and Optimization: Challenges and Opportunities.}}
{This work was also supported by National Science Foundation CAREER DMS-2143915 and by U.S. Department of Energy/Office of Electricity Advanced Grid Modeling program.}
{\textcolor{black}{This research used resources of the Oak Ridge Leadership Computing Facility, which is a
DOE Office of Science User Facility supported under Contract DE-AC05-00OR22725.}}
}
                
\bibliographystyle{abbrvnat}
\bibliography{bib}

\end{document}